\documentclass{tlp}

\usepackage[T1]{fontenc}
\usepackage{soul}
\usepackage{amsmath,amssymb}
\usepackage{amsfonts}
\usepackage{stmaryrd}
\usepackage{mathrsfs}
\usepackage{xspace}
\usepackage{tikz}
\usetikzlibrary{shapes,positioning,intersections}
\usepackage{orcidlink}

\usepackage{hyperref}
\hypersetup{colorlinks=true,allcolors=blue}

\makeatletter
\def\@journal{%
 \vbox to 5.6\p@{\noindent\parbox[t]{4.8in}{\normalfont\affilsize\rmfamily
   {\itshape Under consideration in
    Theory and Practice of Logic Programming (TPLP)\/}\\[2.5\p@]
    \ }%
 \vss}%
}
\let\@j@urnal\@journal
\makeatother

\newtheorem{theorem}{Theorem}[section]
\newtheorem{proposition}[theorem]{Proposition}
\newtheorem{lemma}[theorem]{Lemma}
\newtheorem{example}[theorem]{Example}
\newtheorem{definition}[theorem]{Definition}

\newcommand{\overbar}[1]{\mkern 1.5mu\overline{\mkern-1.5mu#1\mkern-1.5mu}\mkern 1.5mu}

\newcommand{\partof}{\mathscr{P}}
\newcommand{\fun}{\rightarrow}

\newcommand{\Nat}{\mathbb{N}}
\newcommand{\fwp}{\partof_f}
\newcommand{\supp}[1]{\llfloor #1 \rrfloor}
\newcommand{\mwp}{\partof_m}
\newcommand{\emptymulti}{\multil \multir}
\newcommand{\multil}{\{\!\!\{}
\newcommand{\multir}{\}\!\!\}}
\newcommand{\multisum}{\uplus}
\newcommand{\Multisum}{\biguplus}

\newcommand{\Var}{\mathcal{V}}
\newcommand{\Subst}{\mathit{Subst}}
\newcommand{\Isubst}{\mathit{ISubst}}
\newcommand{\Ren}{\mathit{Ren}}
\newcommand{\mgu}{\mathsf{mgu}}

\newcommand{\match}{\mathsf{match}}
\newcommand{\matchlp}{\mathsf{match}_\lp}
\newcommand{\matchan}{\mathsf{match}_\an}
\newcommand{\matchshl}{\mathsf{match}_\shl}
\newcommand{\matchaprime}{\mathsf{match}'_\an}
\newcommand{\occ}{\mathit{occ}}
\newcommand{\rng}{\mathrm{rng}}
\newcommand{\dom}{\mathrm{dom}}
\newcommand{\vars}{\mathrm{vars}}

\newcommand{\Sharing}{\mathtt{Sharing}}

\newcommand{\Linp}{{\mathtt{ShLin}^{\omega}}}
\newcommand{\lp}{\mathrm{\omega}}

\newcommand{\ShLinp}{{\mathtt{ShLin}^2}}
\newcommand{\Andysh}{{\mathit{Sg}^2}}
\newcommand{\an}{\mathrm{2}}
\newcommand{\andybin}{\uplus}
\newcommand{\Andybin}{\biguplus}
\newcommand{\downclo}{{\mathop{\downarrow}}}

\newcommand{\Lin}{\mathtt{Lin}}
\newcommand{\ShLin}{\mathtt{ShLin}}
\newcommand{\shl}{\mathit{sl}}

\newcommand{\wrt}{w.r.t.~}
\newcommand{\ie}{i.e., }

\newcommand{\ra}{\rightarrow}

\newcommand{\Ra}{\Rightarrow}

\newcommand{\member}{\mathit{member}}

\begin{document}


\title{Optimal matching for sharing and linearity analysis}

\begin{authgrp}
  \author{\sn{Amato} \gn{Gianluca}
  \orcidlink{0000-0002-6214-5198}}
  \affiliation{University of Chieti--Pescara\\  	
  	\email{gianluca.amato@unich.it}}
  \author{\sn{Scozzari} \gn{Francesca}
  \orcidlink{0000-0002-2105-4855}}
  \affiliation{University of Chieti--Pescara\\	
  \email{francesca.scozzari@unich.it}}
\end{authgrp}


\maketitle

\begin{abstract}
  Static analysis of logic programs by abstract interpretation requires
  designing abstract operators which mimic the concrete ones, such as
  unification, renaming and projection. In the case of goal-driven
  analysis, where goal-dependent semantics are used, we also need a backward-unification
  operator, typically implemented through matching. In this paper we study
  the problem of deriving optimal abstract matching operators for
  sharing and linearity properties. We provide an optimal operator
  for matching in the domain $\Linp$, which can be easily instantiated to
  derive optimal operators for the domains $\ShLinp$ by Andy King and
  the reduced product $\Sharing \times \Lin$.
\end{abstract}


\section{Introduction}

In the field of static analysis of logic programs, sharing information is one of the most interesting and widely used property.
The goal of sharing analysis is to detect sets of variables which share a common variable.  For instance, in the substitution $\{x/f(z,a),y/g(z) \}$ the
variables $x$ and $y$ share the common variable $z$. Sharing may also track and infer groundness in the same way as the \texttt{Def} domain \citep{delaBandaH93-ilps,ArmstrongMSS98-scp}.
 Typical applications of
sharing analysis are in the
fields of optimization of unification \citep{Sondergaard86-esop} and
parallelization of logic programs \citep{HermenegildoR95-jlp}.

It is widely recognized that the pioneering abstract domain $\Sharing$ \citep{Langen90-phd,JacobsL92-jlp} is not very precise, so
that it is often combined with other domains for tracking freeness, linearity, groundness or structural information (see \citeN{BagnaraZH05-tplp,CodishMBB+95-toplas} for comparative
evaluations).

Any domain for static analysis of logic programs must be equipped with four
standard operators: renaming, projection, union and unification. The theory of
abstract interpretation \citep{CousotC79-popl,CousotC92-jlp} ensures the
existence of the optimal (best correct) abstract operator for each concrete
operator. Nevertheless, while finding optimal operators for
renaming, projection and union is trivial most of the time, devising an optimal
abstract unification is much harder.


\cite{AmatoS10-tplp, AmatoS2014-tplp} have proposed a new (infinite) domain $\Linp$ which
precisely represents the interaction between  sharing and linearity properties,
while discharging other structural, irrelevant information. All the abstract
operators for $\Linp$ are shown to be optimal. From $\Linp$ the authors derive, for the first
time, optimal abstract unification for well-known domains combining sharing and
linearity, such as $\ShLinp$ \citep{King94-esop} and $\Sharing \times \Lin$
\citep{MuthukumarH92-jlp}.

In this paper we extend the $\Linp$ framework to the case of goal-dependent analysis. In this setting, the unification operator is used twice (see Figure~\ref{fig:unif}):
\begin{description}
	\item[forward unification:] performs parameter passing by unifying the goal and
	the \emph{call substitution} with the head of the chosen clause. The
	result is called \emph{entry substitution}.
	\item[backward unification:] propagates back to the goal the \emph{exit
		substitution} (that is, the result of the sub-computation), obtaining the
	\emph{answer substitution}\footnote{We follow
		\citeN{CortesiFW96-jlp} for the terminology of \emph{forward} and \emph{backward
			unification}. \citeN{Bruynooghe91-jlp} and \citeN{HansW92-tr} use \emph{procedure entry}
		and \emph{procedure exit}. \citeN{MuthukumarH91-iclp} use  \emph{call\_to\_entry}
		and \emph{exit\_to\_success}.}.
\end{description}

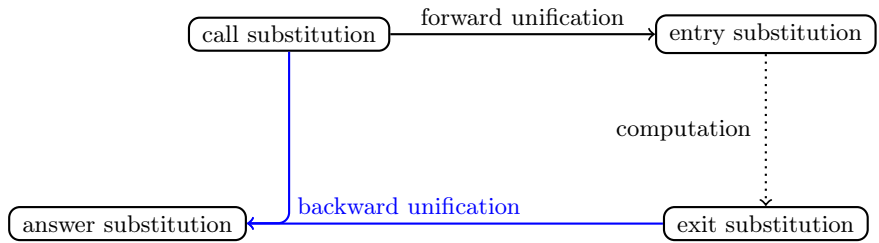
\begin{figure}
\begin{tikzpicture}[rounded corners,thick,rectangle,draw,->,inner xsep=5pt]
 \node[draw] (call) {call substitution};
  \node[draw, right=3.5cm of call] (entry) {entry substitution};
  \node[draw, below=2cm of entry] (exit) {exit substitution};
  \node[draw, left=5.5cm of exit] (answer) {answer substitution};
  \draw (call) edge node[above] {forward unification} (entry);
  \draw (entry) edge[dotted] node[left] {computation} (exit);
  \draw (exit) edge[color=blue] node[above,pos=0.61,color=blue] {backward unification} (answer);
  \node[inner sep=0pt] (inter) at (intersection of call--answer-|call and exit--answer) {};
  \node[inner sep=0pt,above=2pt of inter] (intera) {};
  \draw (call) -- ([yshift=0.2pt]inter.center) [color=blue]-> (answer);
\end{tikzpicture}
\caption{\label{fig:unif}The role of forward and backward unification in goal-dependent analysis.}
\end{figure}

Despite its name, backward unification may be implemented through \emph{matching}, exploiting the property that the  exit substitution is always more instantiated than the call substitution. Analyses with matching are strictly more precise than analyses which do not use matching (see \citeN{Bruynooghe91-jlp} and \citeN{AmatoS09-tplp} for a thorough discussion of the problem). This idea has been implemented in real abstract interpreters such as GAIA \citep{LeCharlierVH94-toplas} and PLAI \citep{MuthukumarH92-jlp}.




However, except for \cite{AmatoS09-tplp}, none of the papers which are based on matching \citep{Bruynooghe91-jlp,LeCharlierVH94-toplas,HansW92-tr,MuthukumarH92-jlp,King95-tr} has ever proved optimality of the proposed abstract operators. In particular, there is no known optimal matching operator for any domain combining sharing and linearity.

The lack of optimal operators brings two kinds of disadvantages: first, the
analysis obviously loses precision when using suboptimal abstract
operators; second, computing approximated abstract objects can lead to
a speed-down of the analysis. The latter is typical of sharing
analysis, where abstract domains are usually defined in such a way
that, the less information we have, the more abstract objects are
complex. This is not the case for other kind of analyses, such as
groundness analysis, where the complexity of abstract objects may grow
accordingly to the amount of groundness information they encode. Moreover,
knowing the optimal abstract operator, even if we do
not plan to implement it, is useful to understand the potentiality and
limits of the abstract domain in use, and to guide the search for a more precise (or more
efficient) domain.

For this reason, in this paper we define a matching operator for $\Linp$ and prove its optimality. Moreover, from this operator, we derive, for the first time, optimal matching operators for domains combining sharing and linearity information, such as $\ShLinp$ and $\Sharing \times \Lin$.

\section{Notations}
\label{sec:notation}


We fix a first order signature that includes a constant symbol and a function symbol of arity at least two, otherwise every term has at most one variable, and the structure of terms is trivial (we need this assumption in the proofs of optimality). The signature also includes a denumerable set of variables $\Var$. Given a term or other syntactic object $o$, we denote by $\vars(o)$ the set of variables occurring in $o$.
Given a set $A$, we denote by $\partof(A)$ the powerset of $A$ and by $\fwp(A)$ the set of finite subsets of $A$.


\subsection{Multisets}
\label{sec:multiset}

A \emph{multiset} is a set where repetitions are allowed. We denote by $\multil x_1, \ldots, x_m \multir$  a multiset, where $x_1, \ldots, x_m$ is a sequence with (possible) repetitions, and by $\emptymulti$ the empty multiset. We will often use the polynomial notation $v_1^{i_1} \ldots v_n^{i_n}$, where $v_1, \ldots, v_n$ is a sequence without repetitions, to denote a multiset $A$ whose element $v_j$ appears $i_j$ times. The set $\{v_j \mid i_j > 0\}$  is called the \emph{support} of $A$ and is denoted by $\supp{A}$. We also use the functional notation $A: \{v_1, \ldots, v_n\} \fun \Nat$, where $A(v_j)=i_j$.

In this paper, we only consider multisets whose support is \emph{finite}. We denote with $\mwp(X)$ the set of all the multisets whose support is \emph{any finite subset} of $X$. For example, $a^3c^5$ and $a^3b^2c^1$ are elements of $\mwp(\{a,b,c\})$.
%
%
The fundamental operation for multisets is the \emph{sum}, defined as
\[
  A \multisum B = \lambda v \in \supp{A} \cup \supp{B}. A(v)+B(v)
  \enspace .
\]
Note that we also use $\multisum$ to denote disjoint union for standard sets. The context will allow us to discern the proper meaning of $\uplus$. Given a multiset $A$ and $X \subseteq \supp{A}$, the \emph{restriction} of $A$ over $X$, denoted by $A_{|X}$, is the only multiset $B$ such that $\supp{B}=X$ and $B(v)=A(v)$ for each $v \in X$.



\subsection{The domain of existential substitutions}
\label{sec:existential}

We work in the framework of existential substitutions \citep{AmatoS09-tplp}, which allows us to simplify those semantic definitions which are heavily based on renaming  apart objects and to avoid variable clashes. In this framework all the  details  concerning renamings are moved to the inner level of the semantic  domain, where they are more easily manageable. We briefly recall the basic definitions of the domain.

The set of substitutions, idempotent substitutions, and renamings are denoted by $\Subst$, $\Isubst$, and $\Ren$ respectively. Given $\theta_1, \theta_2 \in \Subst$ and $U \in \fwp(\Var)$, the preorder $\preceq_U$ is defined as follows:
\[
  \theta_1 \preceq_U \theta_2 \iff \exists  \delta \in \Subst. \forall x \in U.\ \theta_1(x)=\delta(\theta_2(x)) \enspace.
\]

The notation $\theta_1 \preceq_U \theta_2$ states that $\theta_1$ is an instance of $\theta_2$ \wrt the variables in $U$. The equivalence relation induced by the preorder $\preceq_U$ is given by:
\[
  \theta_1 \sim_U \theta_2 \iff \exists \rho \in \Ren.
  \forall x \in U.\ \theta_1(x)=\rho(\theta_2(x)) \enspace .
\]


Let $\Isubst_{\sim_U}$ be the quotient set of $\Isubst$ \wrt $\sim_U$. The domain $\Isubst_\sim$ of \emph{existential substitutions} is defined as the disjoint union of all the $\Isubst_{\sim_U}$ for $U\in \fwp(\Var)$, namely:
\[
  \Isubst_\sim = \biguplus_{U\in \fwp(\Var)} \Isubst_{\sim_U} \enspace .
\]
In the following we write $[\theta]_{U}$ for the equivalence class of $\theta$ \wrt $\sim_U$.  The partial order $\preceq$ over $\Isubst_\sim$ is given by:
\[
  [\theta]_U \preceq [\theta']_V \iff U \supseteq V \wedge \theta \preceq_V \theta' \enspace .
\]
Intuitively, $[\theta]_U \preceq [\theta']_V$ means that $\theta$ is an instance of $\theta'$ \wrt the variables in $V$, provided that they are all variables of interest of $\theta$.

To ease notation, we often omit braces from the sets of variables of interest when they are given extensionally. So we write $[\theta]_{x,y}$ or  $[\theta]_{xy}$ instead of $[\theta]_{\{x,y\}}$ and $\sim_{x,y,z}$ instead of $\sim_{\{x,y,z\}}$.  When the set of variables of interest is clear from the context or when it is not relevant, it will be omitted. Finally, we omit the braces which enclose the bindings of a substitution when the latter occurs inside an equivalence class, \ie we write $[x/y]_U$ instead of $[\{x/y\}]_U$.

\subsubsection{Unification}
\label{sec:unification}

Given $U,V \in \fwp(\Var)$, $[\theta_1]_U, [\theta_2]_V \in \Isubst_{\sim}$, the most general unifier between these two classes is defined as the mgu of suitably chosen representatives, where variables not of interest are renamed apart. In formulas:
\begin{equation}
  \label{eq:mgu}
  \mgu([\theta_1]_U,[\theta_2]_V)= [\mgu(\theta'_1, \theta'_2)]_{U \cup V} \enspace,
\end{equation}
where $\theta_1 \sim_U \theta'_1 \in \Isubst$, $\theta_2 \sim_V \theta'_2 \in \Isubst$ and $(U \cup \vars(\theta'_1)) \cap (V \cup \vars(\theta'_2)) \subseteq U \cap V$. The last condition is needed to avoid variable clashes between the chosen representatives $\theta'_1$ and $\theta'_2$. It turns out that $\mgu$ is the greatest lower bound of $\Isubst_{\sim}$ ordered by $\preceq$.


A different version of unification is obtained when one of the two arguments is an existential substitution and the other one is a standard substitution. In this case, the latter argument may be viewed as an existential substitution where all the variables are of interest:
\begin{equation}
  \label{eq:mixmgu}
  \mgu([\theta]_U,\delta)= \mgu([\theta]_U,[\delta]_{\vars(\delta)}) \enspace .
\end{equation}
Note that deriving the general unification in \eqref{eq:mgu} from the special case in \eqref{eq:mixmgu} is not possible. This is because there are elements in $\Isubst_\sim$ which cannot be obtained as $[\delta]_{\vars(\delta)}$ for any $\delta \in \Isubst$.

This is the form of unification which is better suited for analysis of logic programs, where existential substitutions are the denotations of programs while standard substitutions are the result of unification between goals and heads of clauses. Devising optimal abstract operators for \eqref{eq:mixmgu} in different abstract domains is the topic of \cite{AmatoS10-tplp}.

\subsubsection{Matching}
\label{sec:matching}
Given $U_1,U_2 \in \fwp(\Var)$, $[\theta_1]_{U_1}  \in \Isubst_{\sim}$ and $[\theta_2]_{U_2}  \in \Isubst$, the matching of $[\theta_1]_{U_1}$ with $[\theta_2]_{U_2}$  \citep{AmatoS09-tplp} is defined in the same way as unification, as soon as none of the variables in $U_1$ get instantiated in the result. If this is not the case, the matching is undefined.

\begin{definition}[Matching]
  Given $[\theta_1]_{U_1}, [\theta_2]_{U_2} \in \Isubst_{\sim}$ we have that
  \begin{equation}
    \label{eq:matching}
    \match([\theta_1]_{U_1},[\theta_2]_{U_2})=\begin{cases}
      \mgu([\theta_1]_{U_1},[\theta_2]_{U_2}) & \text{if }  \theta_1 \preceq_{U_1
      \cap U_2} \theta_2, \\
      \text{undefined} & \text{otherwise.}
    \end{cases}
  \end{equation}
\end{definition}
Note that the condition $\theta_1 \preceq_{U_1 \cap U_2} \theta_2$ is equivalent to  $[\theta_1]_{U_1}= \mgu([\theta_1]_{U_1},[\theta_2]_{U_2})_{|U_1}$ \citep{AmatoS09-tplp}.

\begin{example}
  \label{ex:1}
  If we unify $[\theta_1]_{x,y}=[x/a, y/b]_{x,y}$ with $[\theta_2]_{y,z}=[z/r(y)]_{y,z}$, we obtain $[\theta_3]_{x,y,z}=[x/a, y/b, z/r(b)]_{x,y,z}$. Note that the variables $y$ and $z$ in $\theta_3$ are instantiated w.r.t. $\theta_2$, therefore $\match([\theta_2]_{y,z}, [\theta_1]_{x,y})$ is undefined. However, $x$ and $y$ in $\theta_3$ are not instantiated w.r.t. $\theta_1$, therefore $\match([\theta_1]_{x,y},[\theta_2]_{y,z})=[\theta_3]_{x,y,z}$.
\end{example}

Most of the time, when matching is applied in goal-dependent analysis of logic programs, we have that $U_1 \subseteq U_2$. This is because $U_1$ is the set of variables in a clause, while $U_2$ contains both the variables in the clause and in the call substitution. Nonetheless, we study here the general case, so that it can be applied in any framework.

\subsubsection{Other operations}
\label{sec:otherops}

Given $V \in \fwp(\Var)$ and $[\theta]_{U} \in \Isubst_\sim$, we denote by  $([\theta]_U)_{|V}$ the \emph{projection} of $[\theta]_U$ on the set of variables $V$, defined as:
\begin{equation}
  ([\theta]_{U})_{|V} = [\theta]_{U \cap V} \enspace .
\end{equation}

\subsection{Abstract interpretation}

Given two sets $C$ and $A$ of concrete and abstract objects
respectively, an \emph{abstract interpretation} \citep{CousotC92-jlc} is
given by an approximation relation $\rightslice \subseteq A \times C$.
When $a \rightslice c$ holds, this means that $a$ is a correct
abstraction of $c$. We are interested in the case when
$(A,\leq_A)$ is a poset and $a \leq_A a'$ means that $a$ is more
precise than $a'$.  In this case we require that, if $a \rightslice c$
and $a \leq_A a'$, then $a' \rightslice c$, too. In more detail, we
require what \citeN{CousotC92-jlc} call the \emph{existence of the best
  abstract approximation assumption}, \ie the existence of a map
$\alpha: C \ra A$ such that for all $a \in A, c \in C$, it holds that
$a \rightslice c \iff \alpha(c) \leq_A a$. The map $\alpha$ is called
the \emph{abstraction function} and maps each $c$ to its best
approximation in $A$.

Given a (possibly partial) function $f: C \ra C$, we say that
$\tilde{f}: A \ra A$ is a correct abstraction of $f$, and write
$\tilde{f} \rightslice f$, whenever
\[
a \rightslice c  \Rightarrow \tilde{f}(a)
  \rightslice f(c) \enspace ,
\]
assuming that $\tilde{f}(a) \rightslice f(c)$ is true whenever $f(c)$
is not defined.  We say that $\tilde{f}: A \ra A$ is the
\emph{optimal} abstraction of $f$ when it is the best correct
approximation of $f$, \ie when $\tilde{f} \rightslice f$ and
\[
\forall f':A \ra A.\  f' \rightslice f
  \Rightarrow \tilde{f} \leq_{A \ra A} f' \enspace .
\]
In some cases, we prefer to deal with a stronger framework, in which
the domain $C$ is also endowed with a partial order $\leq_C$ and
$\alpha: C \ra A$ is a left adjoint to $\gamma: A \ra C$, \ie
\[
  \forall c \in C. \forall a \in A. \alpha(c) \leq_A a \iff c \leq_C
  \gamma(a) \enspace .
\]
The pair $\langle \alpha, \gamma \rangle$ is called a \emph{Galois
  connection}.  In particular, we will only consider the case of
\emph{Galois insertions}, which are Galois connections such that
$\alpha \circ \gamma$ is the identity map. If $\langle \alpha, \gamma
\rangle$ is a Galois insertion and $f: C \ra C$ is a monotone map, the
optimal abstraction $\tilde{f}$ always exists, and it is definable as
$\tilde{f}= \alpha \circ f \circ \gamma$.

\section[Abstract matching over ShLinw]{Abstract matching over $\Linp$}
\label{sec:shlinomega}

The domain $\Linp$ \citep{AmatoS10-tplp} generalizes $\Sharing$ by
recording  multiplicity of variables in sharing groups.
For example, the
substitution $\theta=\{x/s(u,v), y/g(u,u,u), z/v\}$ is abstracted on
$\Sharing$ into $\{uxy,vxz\}$, where the sharing group $uxy$  means that
$\theta(u)$, $\theta(x)$, and $\theta(y)$ share a common variable, namely $u$.
In $\Linp$ the same substitution would be abstracted as $\{uxy^3,vxz\}$,
with the additional information that the variable $u$ occurs three times
in $\theta(y)$.
For the sake of completeness, in the following section we recall the basic definitions.

\subsection[The domain ShLin\textasciicircum w]{The domain $\Linp$}

We call \emph{$\omega$-sharing group} a
multiset of variables, \ie an element of $\mwp(\Var)$.
Given a substitution $\theta$ and a variable $v \in \Var$, we denote
by $\theta^{-1}(v)$ the $\omega$-sharing group $\lambda w \in \Var.
  \occ(v,\theta(w))$, which maps each variable $w$ to the number of
occurrences of $v$ in $\theta(w)$.

Given a set of variables $U$ and a set of $\omega$-sharing groups $S
  \subseteq \mwp(U)$, we say that $[S]_U$ \emph{correctly
  approximates} a substitution $[\theta]_W$ if $U=W$ and, for each $v
  \in \Var$, $\theta^{-1}(v)_{|U} \in S$. We write $[S]_U \rightslice
  [\theta]_W$ to mean that $[S]_U$ correctly approximates $[\theta]_W$.
Therefore, $[S]_U \rightslice [\theta]_U$ when $S$ contains all the
$\omega$-sharing groups in $\theta$, restricted to the
variables in $U$.

\begin{definition}[$\Linp$]
  The domain $\Linp$ is defined as
  \begin{equation}
    \Linp =\{ [S]_U \mid U \in \fwp(\Var), S \subseteq \mwp(U),
    S \neq \emptyset \Rightarrow \emptymulti \in S\} \enspace ,
  \end{equation}
  and ordered by $[S_1]_{U_1} \leq_\omega [S_2]_{U_2}$ iff $U_1=U_2$ and
  $S_1 \subseteq S_2$.
\end{definition}
The existence of the empty multiset, when $S$ is not empty, is required in
order to have a surjective abstraction function.

In order to ease the notation, we write $[\{\emptymulti, B_1,\ldots, B_n\}]_U$
as  $[B_1,\ldots, B_n]_U$ by omitting the braces and the empty
multiset, and variables in each $\omega$-sharing group are sorted lexicographically. Moreover, if $X \in \Linp$, we write
$B \in X$ in place of $X=[S]_U \wedge B \in S$. Analogously, if $S' \in
  \Linp$, we write $S' \subseteq X$ in place of $X=[S]_U \wedge S' \subseteq S$.
The best correct abstraction
of a substitution $[\theta]_U$ is
\begin{equation}
  \alpha_\omega([\theta]_U)=[\{ \theta^{-1}(v)_{|U} \mid v \in \Var \}]_U
  \enspace .
\end{equation}

\begin{example}
  \label{ex:2}
  Given $\theta=\{x/s(y,u,y), z/s(u,u), v/u\}$ and $U=\{w,x,y,z\}$, we have
  $\theta^{-1}(u)=uvxz^2$, $\theta^{-1}(y)= x^2y$,
  $\theta^{-1}(z)=\theta^{-1}(v)=\theta^{-1}(x)=\emptymulti$ and
  $\theta^{-1}(o)=o$ for all the other variables ($w$ included).
  Projecting over $U$ we obtain
  $\alpha_\omega([\theta]_U)=[x^2y,xz^2,w]_U$.
\end{example}

\subsection{Matching operator}
In the matching operation, when $[\theta_1]_{U_1}$ is matched with $[\theta_2]_{U_2}$, sharing groups in
$\theta_1$ and $\theta_2$ are joined together in the resulting  substitution.
However, not all combinations are allowed.
Assume $\alpha_\lp([\theta_i]_{U_i})=[S_i]_{U_i}$ for $i \in
  \{1,2\}$. If $\match([\theta_1]_{U_1},[\theta_2]_{U_2})$ is defined,
$\theta_1$ will not be further instantiated and thus
$\alpha_{\lp}(\match([\theta_1]_{U_1}, [\theta_2]_{U_2})_{|U_1} ) \subseteq
  S_1$.
Moreover, the sharing groups in $S_2$ which do not contain any variable in
$U_1$ are not affected by
the unification, since the corresponding existential variable does not
appear in $\theta_2(v)$ for any $v \in U_1$.

\begin{example}
  \label{ex:3}
  Let $\theta_1 = \{ x/r(w_1,w_2,w_2,w_3,w_3), y/a, z/r(w_1) \}$ with
  $U_1=\{x,y,z\}$ and  $\theta_2 = \{ x/r(w_4, w_5, w_6,w_8,w_8), u/r(w_4,w_7),
    v/r(w_7,w_8)\}$  with  $U_2=\{u,v,x\}$. We have that
  \begin{multline*}
    [\theta]_U= \match([\theta_1]_{U_1}, [\theta_2]_{U_2}) = [u/r(w_1,w_7),
    v/r(w_7,w_3), \\
    x/r(w_1,w_2,w_2,w_3,w_3), y/a, z/r(w_1)]_U \enspace ,
  \end{multline*}
  with $U=\{u,v,x,y,z\}$. At the abstract level, we have $[S_1]_{U_1}=
    \alpha_\lp([\theta_1]_{U_1})= [x^2, xz]_{U_1}$, $[S_2]_{U_2}=
    \alpha_\lp([\theta_2]_{U_2})= [uv,ux,vx^2,x]_{U_2}$ and $[S]_U =
    \alpha_\lp([\theta]_U)= [uv, uxz, vx^2, x^2]_U$.

  Note that the new sharing group $uxz$ has the property that its restriction
  to $U_1$ is in $S_1$. More generally, if we abstract $\theta_1$ \wrt the variables in $U_1$ we get $\alpha_\lp([\theta]_{U_1}) = [x^2, xz]_{U_1}$. This is equal to $[S_1]_{U_1}$, showing that no new sharing group has been introduced by the matching operation relatively to the variables in $U_1$. Moreover, the sharing group $uv$, which does not contain variables in $U_1$, is brought unchanged from $S_2$ to $S$.
\end{example}

Following the idea presented above, we may design an abstract matching
operation for the domain $\Linp$.

\begin{definition}[Matching over $\Linp$]
  \label{def:matchlp}
  Given $[S_1]_{U_1}, [S_2]_{U_2} \in \Linp$ we define
  \[
    \matchlp([S_1]_{U_1},[S_2]_{U_2})=\big[S'_2 \cup
      \big\{ X \in \mwp(U_1 \cup U_2) \mid \\
      X_{|U_1} \in S_1 \wedge X_{|U_2} \in (S''_2)^* \big\} \big]_{U_1
      \cup U_2}
  \]
  where
  \[
    S'_2=\{ B \in S_2 \mid B_{|U_1} = \emptyset \} \qquad S''_2=S_2
    \setminus S'_2 \enspace \qquad
    S^*=\left\{ \multisum \mathcal S \mid
    \mathcal S \in \mwp(S) \right\} \enspace .
  \]
\end{definition}
%
%
We now show an example of computing the matching over $\Linp$.

\begin{example}
  \label{ex:4}
	Consider $[S_1]_{U_1}=[x^2, xz]_{xyz}$ and $[S_2]_{U_2}=[uv,ux,vx^2,x]_{uvx}$ as in Example~\ref{ex:3}.
	We show that, from the definition of $\matchlp$, it holds:
	\[\matchlp([S_1]_{U_1},[S_2]_{U_2}) = [uv, uxz, xz, u^2x^2, ux^2, vx^2, x^2]_{uvxyz} \enspace .\]

	First, we have that $S'_2 = \{uv\}$ and $S''_2 = \{ux, vx^2, x\}$.
	Apart from $uv$, which directly comes from $S'_2$, all the other $\omega$-sharing groups in the result may be obtained by choosing a multiset $\mathcal M$ of sharing groups in $S''_2$ and summing them together, obtaining $\uplus \mathcal M$. Then, we consider if it is possible to add to $\uplus \mathcal M$ some occurrences of the variables $y$ and $z$, for instance $n$ occurrences of $y$ and $m$ of $z$, in such a way that $(\uplus \mathcal M) \uplus \multil y^{n} z^{m} \multir$ restricted to $U_1$ is a sharing group in $S_1$.

	We start by considering $\mathcal M $ with a single $\omega$-sharing group.
	\begin{itemize}
	\item If $\mathcal M = \multil ux \multir$, then $\uplus \mathcal  M = ux$ but $(\uplus \mathcal X)|_{U_1} = x \notin S_1$. However, we can add the variable $z$ to get $uxz \in S_1$, hence $uxz$ is an $\omega$-sharing group in the result of $\match_\lp$.

	\item If $\mathcal M = \multil vx^2 \multir$, then $\uplus \mathcal M = vx^2$ and $(\uplus \mathcal  M)|_{U_1} = x^2$ which is already an element of $S_1$. Therefore, $vx^2$ is in the result of $\match_\lp$.

	\item If $\mathcal M = \multil x \multir$, then $\uplus \mathcal M = x$, but $(\uplus \mathcal M)|_{U_1} = x \notin S_1$. However, we can add the variable $z$ to get $xz \in S_1$, hence $xz$ is in the result of $\match_\lp$.

	\end{itemize}

	We now consider the cases when $\mathcal M$ has two (possibly equal) elements. Note that if $\mathcal M$ contains $vx^2$, it cannot contain anything else, otherwise $\uplus \mathcal M$ would contain at least three occurrences of $x$, and no sharing group in $S_1$ could be matched. Therefore, the only choices are:
	\begin{itemize}
		\item if $\mathcal M = \multil ux, ux \multir$, then $\uplus \mathcal M = u^2x^2$ and $(\uplus \mathcal  M)|_{U_1} = x^2$
		which is already in $S_1$, hence $u^2x^2$ is in the result of $\match_\lp$;
		\item if $\mathcal M = \multil x ,x \multir$, then $\uplus \mathcal M = x^2$ and $(\uplus \mathcal  M)|_{U_1} = x^2$
		which is already in $S_1$, hence $x^2$ is in the result of $\match_\lp$;
		\item if $\mathcal M = \multil ux ,x \multir$, then $\uplus \mathcal M = ux^2$ and $(\uplus \mathcal  M)|_{U_1} = x^2$
		which is already in $S_1$, hence $ux^2$ is in the result of $\match_\lp$.
	\end{itemize}
	In theory, we should also consider the case when $\mathcal M$ has more than two elements, but in the example this would not lead to new results, for the same reason why $vx^2$ may only be used alone.


  We will prove that $\matchlp$ is correct, therefore comes at no surprise
  that
  \[
    \alpha_\lp(\match([\theta_1]_{U_1},[\theta_2]_{U_2}])))= [uv, uxz, vx^2,
    x^2]_U
    \leq [uv, uxz, xz, u^2x^2, ux^2, vx^2, x^2]_{U} \enspace .
  \]
  The lack of equality means that $\matchlp$ is not ($\alpha$-)complete \citep{GiacobazziRS00-jacm,AmatoS11-fi}. We will show
  later that it is optimal.
\end{example}

We try to give the intuition behind the definition of $\match$, especially in the context of the backward unification. First note
that the operator is additive on the first argument, namely:

\[\matchlp([S_1]_{U_1},[S_2]_{U_2}) = \left[\bigcup_{B\in S_1}
    \matchlp([\{B\}]_{U_1},[S_2]_{U_2}) \right]_{U_1\cup U_2}\]

This immediately implies that we can reason on matching considering one
sharing group at a time. Given a sharing group $B\in S_1$, which represents
the exit substitution from a sub-computation, we try to guess which
of the sharing groups in $S_2$ are part of the entry substitution
which has generated the sub-computation ended with $B$. In the simple case
when $U_1 \supseteq U_2$, we simply check that $B$
can be generated by the sharing groups in $S_2$, that is, there exists
$\mathcal S \in \mwp(S_2)$ such that $B_{|U_2} = \multisum \mathcal S$.


The difficult case is when $U_2$ contains some variables which are not in $U_1$. These
variables come from the call substitution, they are removed from the abstraction
before entering the sub-computation, and now should be re-introduced as precisely
as possible. In this case, we build a new sharing group $X$ such that
$X_{|U_1}$ coincides with $B$, and $X_{|U_2}$ is generated by $S_2$, namely
there exists $\mathcal S \in \mwp(S_2)$ such that $X_{|U_2} = \multisum
  \mathcal S$.
This condition ensures that we pair each exit substitution $\theta_1$ (in
the concretization of $B$) with some entry substitution $\theta_2$ (in the
concretization of $\mathcal S$) which is less instantiated than $\theta_1$.

Note that, although the abstract unification operator $\mgu_\lp$ defined in
\citeN{AmatoS10-tplp} takes an abstract object and a substitution as inputs,
the operator $\matchlp$  is designed in such a way that both  the
arguments are abstract objects. The reason for this choice is that
these are the variants needed for static analysis of
logic programs.
However, it would be possible to devise a variant of
$\mgu_\lp$ with two abstract arguments and variants of $\matchlp$ with
one abstract argument and a concrete one.

In order to prove correctness of $\matchlp$, we first extend the definition
of $\theta^{-1}$ to the case when it is applied to a sharing group $B$,
elementwise as
\begin{equation}
  \theta^{-1}(v_1^{i_1} \cdots v_n^{i_n} ) = \biguplus \multil
  (\theta^{-1}(v_1))^{i_1}, \dots, (\theta^{-1}(v_n))^{i_n} \multir \enspace .
\end{equation}
It turns out that
\begin{equation}
  (\eta \circ \theta)^{-1}(B)=\theta^{-1}(\eta^{-1}(B))
  \enspace ,
\end{equation}
a result which has been proved in \citeN{AmatoS10-tplp}.

\begin{example}
    Given $\eta=\{x/s(y,u,y), z/s(u,u), v/u\}$ and $\theta = \{ v/a, w/s(x, x) \}$, we have $\eta \circ \theta = \{ v/a, w/s(s(y,u,y), s(y,u,y)), x/s(y,u,y), z/s(u,u))\}$ and  $(\eta \circ \theta)^{-1}(u) = uw^2xz^2$. The same result may be obtained as $\theta^{-1}(\eta^{-1}(u))$, since $\eta^{-1}(u) = uvxz^2$ and
    \begin{multline*}
        \theta^{-1}(uvxz^2) = \multisum \multil \theta^{-1}(u), \theta^{-1}(v),
        \theta^{-1}(x), (\theta^{-1}(z))^2 \multir = \\
        = \multisum  \multil u, \multil \multir, xw^2, (z)^2 \multir = uw^2xz^2 \enspace .
    \end{multline*}
\end{example}

\begin{theorem}[Correctness of $\matchlp$]
  \label{th:matchcorrect}
  $\matchlp$ is correct w.r.t. $\match$.
\end{theorem}
\begin{proof}
  Given $[S_1]_{U_1} \rightslice [\theta_1]_{U_1}$ and $[S_2]_{U_2}
    \rightslice [\theta_2]_{U_2}$ such that
  $\match([\theta_1]_{U_1},[\theta_2]_{U_2})$ is defined, we need to prove
  that
  \[
    \matchlp([S_1]_{U_1},[S_2]_{U_2}) \rightslice
    \match([\theta_1]_{U_1},[\theta_2]_{U_2}) \enspace .
  \]

  Assume without loss of
  generality that $\dom(\theta_1)=U_1$, $\dom(\theta_2)=U_2$ and
  $\vars(\theta_1) \cap \vars(\theta_2) \subseteq U_1 \cap U_2$. In
  particular, this implies $\rng(\theta_1) \cap
    \rng(\theta_2)=\emptyset$.  By hypothesis $\theta_1 \preceq_{U_1 \cap U_2}
    \theta_2$, \ie there exists $\delta \in \Isubst$ such
  that $\theta_1(x)= \delta(\theta_2(x))$ for each $x \in U_1 \cap
    U_2$, $\dom(\delta)=\vars(\theta_2(U_1 \cap U_2))$ and
  $\rng(\delta)=\vars(\theta_1(U_1 \cap U_2))$. We have
  \begin{eqnarray*}
    &&\mgu(\theta_1,\theta_2)\\
    &=&\mgu(\{ \theta_2(x)=\theta_2(\theta_1(x)) \mid x \in U_1 \})
    \circ \theta_2\\
    &&\quad \text{[by assumptions on the $\theta_i$'s]}\\
    &=&\mgu(\{ \theta_2(x)=\theta_1(x) \mid x \in U_1 \})
    \circ \theta_2\\
    &=&\mgu(\{ \theta_2(x)=\theta_1(x) \mid x \in U_1 \cap U_2 \})
    \circ (\theta_1)_{|U_1 \setminus U_2} \circ \theta_2\\
    &&\quad \text{[since $\vars(\theta_2) \cap \vars((\theta_1)_{U_1
              \setminus U_2})=\emptyset]$}\\
    &=&\mgu(\{ \theta_2(x)=\delta(\theta_2(x)) \mid  x \in U_1  \cap U_2\})
    \circ ((\theta_1)_{|U_1 \setminus U_2} \uplus \theta_2)\\
    &=&\delta \circ ((\theta_1)_{|U_1 \setminus U_2} \uplus
    \theta_2)\\
    &&\quad \text{[since $\dom(\delta) \cap \vars((\theta_1)_{U_1
              \setminus U_2})=\emptyset]$}\\
    &=&(\theta_1)_{|U_1 \setminus U_2} \uplus (\delta \circ \theta_2)
    \enspace .
  \end{eqnarray*}
  With an analogous derivation, we obtain
  \[
    \mgu(\theta_1,\theta_2)= \theta_1 \uplus (\delta \circ
      (\theta_2)_{|U_2 \setminus U_1}) \enspace .
  \]
  Now, if $\eta=\mgu(\theta_1,\theta_2)$, we need to prove that, for
  each $v \in \Var$, $\eta^{-1}(v)_{|U_1 \cup U_2} \in
    \matchlp([S_1]_{U_1},[S_2]_{U_2})$. We distinguish several cases.
  \begin{itemize}
    \item $v \notin \rng(\theta_1)$ and $v \notin \rng(\theta_2)$. In
          this case $v \notin \vars(\delta)$ and $\eta^{-1}(v)_{|U_1 \cup
            U_2}=\multil \multir \in  \matchlp([S_1]_{U_1},[S_2]_{U_2})$.
    \item $v \in \rng(\theta_1)$ and $v \notin \vars(\theta_1(U_2))$. In
          this case $v \notin \rng(\theta_2)$ and $v \notin \vars(\delta)$.
          We have $\eta^{-1}(v)_{|U_1 \cup U_2} = \theta_1^{-1}(v)_{|U_1}
            \in S_1 \subseteq  \matchlp([S_1]_{U_1},[S_2]_{U_2})$.
    \item $v \in \rng(\theta_2)$ and $v \notin \vars(\theta_2(U_1))$. In
          this case $v \notin \rng(\theta_1)$ and $v \notin \vars(\delta)$.
          We have $\eta^{-1}(v)_{|U_1 \cup U_2} = \theta_2^{-1}(v)_{|U_2} \in
            S_2$. Since $v \notin \vars(\theta_2(U_1))$, then
          $\theta_2^{-1}(v)_{|U_2} \in S'_2 \subseteq
            \matchlp([S_1]_{U_1},[S_2]_{U_2})$.
    \item $v \in \vars(\theta_2(U_1 \cap U_2))$. In this case $v \notin
            \rng(\theta_1)$ and $v \in \dom(\delta)$, therefore
          $\eta^{-1}(v)_{|U_1 \cup U_2}=\multil \multir$.
    \item $v \in \vars(\theta_1(U_1 \cap U_2))$. Now, $v \in
            \rng(\delta)$ and $\eta^{-1}(v)_{|U_2}=\theta_2^{-1}(
            \delta^{-1}(v))_{|U_2}=\biguplus X$, where $X \in \mwp(S_2)$ and
          each sharing group $B \in X$ is of the form
          $\theta_2^{-1}(w)_{|U_2}$ for some $w \in \supp{\delta^{-1}(v)}$.
          Note that every such $w$ is an element of $\vars(\theta_2(U_1 \cap
              U_2))$, therefore $\theta_2^{-1}(w)_{|U_2} \in S''_2$, $X \in
            \mwp(S''_2)$ and $\biguplus X \in (S_2'')^*$. On the other side, since
          $\eta=\theta_1 \uplus
            (\delta \circ (\theta_2)_{|U_2 \setminus U_1})$, we have
          $\eta^{-1}(v)_{|U_1}=\theta_1^{-1}(v)_{|U_1} \in S_1$. Hence,
          $\eta^{-1}(v)_{|U_1
              \cup U_2} \in       \matchlp([S_1]_{U_1},[S_2]_{U_2})$.
  \end{itemize}
  This concludes the proof.
\end{proof}

Now, we may prove the optimality of $\matchlp$. Actually, we prove a
stronger property, a sort of weak completeness of $\matchlp$, which
will be used later to derive optimality.

\begin{example}
  Consider again the substitutions and sharing groups in Examples~\ref{ex:3}
  and~\ref{ex:4}. Recall that $U_1 = \{x, y, z\}$ and $U_2 = \{ u, v, x \}$.
  We have seen that, although the sharing groups $xz$, $u^2x^2$,
  and $ux^2$ are in $\matchlp([S_1]_{U_1},[S_2]_{U_2})$, they are not in
  $\alpha_\lp(\match([\theta_1]_{U_1}, [\theta_2]_{U_2}))$. If $\matchlp$ were
  optimal, we should be able to find, for each $B \in \{xz, u^2x^2, ux^2\}$,
  a pair of substitutions $\theta_3$ and $\theta_4$ such that
  $[S_1]_{U_1} \rightslice [\theta_3]_{U_1}$, $[S_2]_{U_2} \rightslice
    [\theta_4]_{U_2}$ and $B \in \alpha_\lp(\match([\theta_3]_{U_1},
    [\theta_4]_{U_2})$. Actually, we can do better, keep $\theta_4=\theta_2$
  fixed, and only change $\theta_3$ for different sharing groups.

  Consider $B=xz$. Note that $B_{|U_1}=xz$ and $B_{|U_2}= x = \Multisum \mathcal X$, with
  $\mathcal X= \multil x \multir$. The sharing group $x$ in $\theta_2$ is generated by
  the variables $w_5$ and $w_6$. Consider the substitution $\theta_3 = \{ x/r(a,
    w, a, a, a), y/a, z/w \}$ where
  \begin{itemize}
    \item the binding $x/r(a,w,a, a, a)$ is obtained by the corresponding binding
          in $\theta_2$ replacing $w_5$ with a fresh variable $w$ and all the other
          variables  in the range with the constant symbol $a$;
    \item the bindings $\{ y/a, z/w \}$ are chosen according to $B_{|U_1 \setminus
                U_2}=z$.
  \end{itemize}
  It is immediate to show that $xz \in \alpha_\lp(\mgu([\theta_3]_{U_1},
      [\theta_2]_{U_2}))$ and $[S_1]_{U_1} \rightslice [\theta_3]_{U_1}$.

  As an another example, let us consider $B=u^2x^2$. In this case $B_{|U_1}=x^2
    \in S_1$ and $B_{|U_2} = u^2x^2 = \Multisum \multil ux, ux \multir$. The
  variable which generates the sharing group $ux$ is $w_4$. We proceed as before, and obtain
  $\theta_3 = \{ x/r(r(w,w),a,a,a,a), y/a, z/a\}$. Note that $w_4$ has been
  replaced with $r(w,w)$ since two copies of $ux$ are needed to obtain $u^2x^2$.
  Again  $[S_1]_{U_1} \rightslice [\theta_3]_{U_1}$ and $u^2x^2 \in
    \alpha_\lp(\mgu([\theta_3]_{U_1},  [\theta_2]_{U_2}))$.
\end{example}

We distill the idea presented above in the following result.

\begin{lemma}[Completeness on the second argument]
  \label{lem:matchcompl}
  Given $[S_1]_{U_1} \in \Linp$ and $[\theta_2]_{U_2} \in \Isubst_\sim$, there
  exist $\delta_1,\dots,\delta_n \in  \Isubst_\sim$ such that for all $i\in
    \{1,\dots,n\}$,   $[S_1]_{U_1} \rightslice [\delta_i]_{U_1}$
  and
  \[\matchlp([S_1]_{U_1}, \alpha_\lp([\theta_2]_{U_2}))= \left[\bigcup_ {i\in
        \{1,\ldots,n\}} \alpha_\omega(\match([\delta_i]_{U_1},
      [\theta_2]_{U_2}))\right]_{U_1\cup U_2} \]
\end{lemma}

\begin{proof}
  Since we already know that $\matchlp$ is a correct abstraction of
  $\match$, we only need to prove that, given $[S_1]_{U_1} \in \Linp$ and
  $[\theta_2]_{U_2} \in \Isubst_\sim$, for any
  sharing group $B \in \matchlp([S_1]_{U_1},
    \alpha_\lp([\theta_2]_{U_2}))$, there exists $[\theta_1]_{U_1} \in
    \Isubst_\sim$ such that $[S_1]_{U_1} \rightslice [\theta_1]_{U_1}$
  and $B \in \alpha_\lp(\match([\theta_1]_{U_1},
      [\theta_2]_{U_2}))$.

  In order to ease notation, let $U=U_1 \cup U_2$, $[S_2]_{U_2}=\alpha_\lp([\theta_2]_{U_2})$ and
  $[S]_U=\matchlp([S_1]_{U_1}, [S_2]_{U_2})$. We may choose $\theta_2$
  such that $\dom(\theta_2)=U_2$ without loss of generality. Moreover,
  $S'_2$ and $S''_2$ are given as in the definition
  of abstract matching. We distinguish two cases.

  \begin{description}
    \item[first case)]
      If $B \in S'_2$, there exists $v \in \Var$ such that
      $\theta_2^{-1}(v)_{|U_2}=B$. Let $X=\vars(\theta_2(U_1 \cap U_2))$
      and take $\delta=\{ x/a \mid x \in X \}$. Then $\theta_1=(\delta
        \circ \theta_2)_{|U_1} \uplus \{ x/a \mid x \in U_1 \setminus U_2\}$
      is such that $\theta_1^{-1}(v)_{|U_1}= \multil \multir$ for each
      $v \in \Var$, therefore $[S_1]_{U_1} \rightslice [\theta_1]_{U_1}$.
      Moreover, since $\theta_1 \preceq_{U_1 \cap U_2} \theta_2$, we have
      that $\matchlp([\theta_1]_{U_1}, [\theta_2]_{U_2})$ is defined equal
      to $\mgu([\theta_1]_{U_1}, [\theta_2]_{U_2})$ and
      $\mgu([\theta_1]_{U_1},
        [\theta_2]_{U_2})=[\mgu(\theta_1,\theta_2)]_U$ since
      $\vars(\theta_1)\cap \vars(\theta_2) \subseteq U_1 \cap U_2$.  By
      the proof of Theorem \ref{th:matchcorrect}, we have that
      $\mgu(\theta_1,\theta_2)=(\theta_1)_{|U_1 \setminus U_2} \uplus
        (\delta \circ \theta_2)$.  Since $B_{|U_1}=\multil \multir$, then $v
        \notin X=\vars(\delta)$, and therefore
      $\mgu(\theta_1,\theta_2)^{-1}(v)_{|U}=\theta_2^{-1}(v)_{|U}=B$.
      Hence, $B$ is an $\omega$-sharing group in
      $\alpha_\lp(\match([\theta_1]_{U_1},
          [\theta_2]_{U_2}))$ which is what we wanted to prove.

    \item[second case)] We now assume $B_{|U_1} \in S_1$ and $B_{|U_2}=\biguplus
        X$ with $\mathcal X   \in \mwp(S''_2)$. Then, for each $H \in \supp{\mathcal  X}$, there exists
      $v_H
        \in \Var$ such that $\theta_2^{-1}(v_H)_{|U_2}=H$.  Since $H \cap
        U_1 \neq \multil \multir$ for each $H \in \supp{\mathcal X}$, then $v_H \in
        Y=\vars(\theta_2(U_1 \cap U_2))$. Consider the
        substitutions
      \begin{align*}
        \delta & =\{ v_H/t(\underbrace{v,\ldots,v}_{\text{$\mathcal X(H)$ times}})
        \mid H \in \supp{\mathcal X} \} \uplus \{ y/a \mid y \in Y \text{ and }
        \forall  H \in \supp{\mathcal X}. y \neq v_H\} \enspace ,                  \\
        \eta   & =\{ w/t(\underbrace{v, \ldots,v}_{\text{$B(w)$ times}}) \mid w
        \in U_1 \setminus U_2 \} \enspace ,
      \end{align*}
      for a fresh variable $v$. Let us define $\theta_1=\eta \uplus
        (\delta \circ \theta_2)_{|U_1}$. We want to prove that $[S_1]_{U_1}
        \rightslice [\theta_1]_{U_1}$. Note that
      $\vars(\theta_1(U_1))=\{v\}$, hence we only need to check that
      $\theta^{-1}_1(v)_{|U_1} \in S_1$. We have that
      $\theta^{-1}(v)_{|U_1}=\eta^{-1}(v)_{|U_1} \uplus
        (\theta_2^{-1}\delta^{-1}(v))_{|U_1} = \eta^{-1}(v)_{|U_1} \uplus
        \theta_2^{-1}( \multil
        v_{H}^{X(H)} \mid H \in \supp{\mathcal X} \multir)_{|U_1}=B_{|U_1
        \setminus U_2} \uplus (\biguplus \mathcal X)_{|U_1 \cap U_2}=B_{|U_1}
        \uplus B_{|U_1 \cap U_2}=B_{|U_1}$ which, we know, is an element of
      $S_1$.  Moreover, since $\theta_1 \preceq_{U_1 \cap U_2} \theta_2$ and
      $\vars(\theta_1) \cap \vars(\theta_2) \subseteq U_1 \cap U_2$, we
      have that $\matchlp([\theta_1]_{U_1},[\theta_2]_{U_2})=
        [\mgu(\theta_1,\theta_2)]_U$. If we define $\theta=\mgu(\theta_1,
        \theta_2)$, by looking at the proof of Theorem \ref{th:matchcorrect}
      we have that $\theta=\theta_1 \uplus (\delta \circ (\theta_2)_{|U_2
          \setminus U_1})$ and $\theta=(\theta_1)_{|U_1 \setminus U_2} \uplus
        (\delta \circ \theta_2)$. By the first equality, it is immediate to
      check that $\theta^{-1}(v)_{|U_1} = \theta_1^{-1}(v)_{|U_1}=
        B_{|U_1}$.  By the second equality, $\theta^{-1}(v)_{|U_2}=
        \theta_2^{-1}(\multil v_{H}^{\mathcal X(H)} \mid H \in \supp{\mathcal X} \multir)_{|U_2} =
        \biguplus \mathcal X =B_{|U_2}$.  Therefore, $\theta^{-1}(v)_{|U}=B$.
    \end{description}

    It is worth noting that, although this proof uses a function symbol of arbitrary arity, it may be easily rewritten using only one constant symbol and one function symbol of arity at least two, as required at the beginning of Section~\ref{sec:notation}.
\end{proof}

\begin{theorem}[Optimality of $\matchlp$]
  \label{th:matchlpopt}
  The operation $\matchlp$ is optimal  w.r.t. $\match$.
\end{theorem}
\begin{proof}
  It is enough to prove that for each $[S_1]_{U_1}, [S_2]_{U_2} \in \Linp$ and $B \in \matchlp([S_1]_{U_1},[S_2]_{U_2})$, there are substitutions $\theta_1$, $\theta_2$ such that $[S_1]_{U_1} \rightslice [\theta_1]_{U_1}$,  $[S_2]_{U_21} \rightslice [\theta_2]_{U_2}$ and $B \in \alpha_\lp(\match([\theta_1]_{U_1}, [\theta_2]_{U_2})$.
  Assume $B_{|U_2}=\biguplus \mathcal X$ where $\mathcal X \in \mwp(S''_2)$. Consider a
  substitution $[\theta_2]_{U_2}$ such that $[\supp{\mathcal X}]_{U_2} \leq \alpha_\lp([\theta_2]_{U_2}) \leq [S_2]_{U_2}$.
  It means that, for each $H \in \supp{\mathcal X}$,
  there is $v_H \in \Var$ such that $\theta_2^{-1}(v_H)_{|U_2}=H$. If $\supp{\mathcal X}=\{ H_1, \ldots, H_n \}$, we may define $\theta_2$ as the substitution with $\dom(\theta_2)=U_2$ and, for each $u \in U_2$,
  \[
    \theta_2(u) = t( \underbrace{v_{H_1}, \ldots, v_{H_1}}_{H_1(u) \text{ times}}, \underbrace{v_{H_2}, \ldots, v_{H_2}}_{H_2(u) \text{ times}}, \ldots,  \underbrace{v_{H_n}, \ldots, v_{H_n}}_{H_n(u) \text{ times}} ) \enspace .
  \]

  By Lemma~\ref{lem:matchcompl}, there exists $[\theta_1]_{U_1}$ such that
  $[S_1]_{U_1} \rightslice [\theta_1]_{U_1}$ and $B
    \in \alpha_\lp(\match([\theta_1]_{U_1}, [\theta_2]_{U_2}))$.
  Therefore,  $\matchlp$ is the optimal approximation of $\match$.
\end{proof}

\section[Abstract matching over ShLin\textasciicircum 2]{Abstract matching over $\ShLinp$}
\label{sec:shlin2}

The domain $\Linp$ has been inspired by the domain $\ShLinp$, appeared for the
first time in \cite{King94-esop}.
The novelty of $\ShLinp$ was to embed
linearity information inside the sharing groups, instead of keeping them
separate like it was in $\Sharing \times \Lin$.

\subsection[The domain ShLin\textasciicircum 2]{The domain $\ShLinp$}
Here we recall the main  definitions for the
domain $\ShLinp$, viewed as an abstraction of $\Linp$, following the
presentation  given in \cite{AmatoS10-tplp}.

The idea is to simplify the domain $\Linp$ by only recording whether a
variable in a sharing group is linear or not, but forgetting its actual
multiplicity. Intuitively, we abstract an $\omega$-sharing group by replacing
any exponent equal to or greater than $2$  with a new symbol $\infty$.

A \emph{2-sharing group} is a map $o: \Var \fun \{0,1,\infty\}$ such that
its support $\supp{o}=\{ v \in \Var \mid o(v) \neq 0\}$ is finite.
We use a polynomial notation for 2-sharing groups as for $\omega$-sharing
groups. For instance, $o=xy^\infty z$ denotes
the 2-sharing group whose support is
$\supp{o}=\{x,y,z\}$, such that $o(x)=o(z)=1$ and $o(y)=\infty$. We denote
with $\emptyset$ the 2-sharing group with empty support and by $\Andysh(V)$
the set of 2-sharing groups whose support is a subset of $V$.  Note that in
\cite{King94-esop} the number $2$ is used as an exponent instead of $\infty$, but
we prefer our notation to be coherent with $\omega$-sharing groups.

An $\omega$-sharing group $B$ may be abstracted into the 2-sharing group
$\alpha_\an(B)$ given by
\begin{equation}
  \alpha_\an(B)=\lambda v\in \Var .
  \begin{cases}
    B(v)   & \text{if $B(v)\leq 1$,} \\
    \infty & \text{otherwise.}
  \end{cases}
\end{equation}
For instance, the  $\omega$-sharing groups $xy^{2} z,xy^{3} z,xy^{4} z,xy^{5}
  z,\ldots$ are all abstracted into $xy^\infty z$.

%
%
%
%
%
%

There are operations on 2-sharing groups which correspond to variable
projection and  multiset union. For projection
\begin{equation}
  o_{|V} = \lambda v \in \Var. \begin{cases}
    o(V) & \text{if $v \in V$,} \\
    0    & \text{otherwise,}
  \end{cases}
\end{equation}
while for multiset union
\begin{equation}
  o \andybin o' = \lambda v \in \Var. o(v) \oplus o'(v) \enspace ,
\end{equation}
where $0 \oplus x= x \oplus 0=x$ and $\infty \oplus x=x \oplus
  \infty=1 \oplus 1= \infty$. We will use $\Andybin \multil o_1, \ldots,
  o_n \multir$ for $o_1 \andybin \cdots \andybin o_n$. Given a sharing
group $o$, we also define the \emph{delinearization} operator:
\begin{equation}
  \label{eq:delinearization}
  o^2=o \andybin o  \enspace .
\end{equation}
This operator is extended  pointwise to sets and multisets. The next
proposition shows some properties of these operators.

\begin{proposition}
  \label{prop:abstraction2}
  Given an $\lp$-sharing group $B$, a set of variables $V$ and multiset of $\omega$-sharing groups $\mathcal X$, the following properties hold:
  \begin{enumerate}
    \item $\supp{B}= \supp{\alpha_\an(B)}$
    \item $\alpha_\an(B_{|V})=  \alpha_\an(B)_{|V}$
    \item \label{eq:abstraction1} $\alpha_\an(\biguplus \mathcal X)=
            \biguplus \alpha_\an(\mathcal X)$
    \item $\alpha_\an(B \uplus B) = \alpha_\an(B)^2$
  \end{enumerate}
  where $\alpha_\an(\mathcal X)$ is just the elementwise extension of $\alpha_\an$ to a multiset of $\omega$-sharing groups.
\end{proposition}
\begin{proof}
  Given the $\omega$-sharing group $B$, we have $\supp{\alpha_\an(B)}=
    \{ x \in \Var \mid \alpha_\an(B) \neq 0 \} = \{ x \in \Var \mid B(x) \neq 0 \}
    =  \supp{B}$, which proves Property~1. For the Property~2, given $V \subseteq
    \Var$,  we want to prove that $\alpha_\an(B_{|V})(v)= (\alpha_\an(B)_{|V})(v)$ for
  each $v  \in \Var$. The property is easily proved considering the three cases
  $v \notin V$, $v \in V \setminus \supp{B}$ and $v \in V \cap \supp{B}$.
  Property~3 has been proved in \cite{AmatoS10-tplp}. Property~4 is a trivial
  consequence of Property~3.
\end{proof}

Since we do not want to represent definite non-linearity, we introduce an
order relation over sharing groups as
follows:
\[
  o \leq o' \iff \supp{o}=\supp{o'} \wedge
  \forall x \in \supp{o}.\ o(x) \leq o'(x)  \enspace ,
\]
and we restrict our attention to downward closed sets of sharing
groups. The domain we are interested in is the following:
\[
  \ShLinp=\bigl\{ [T]_U \mid T \in \partof_{\downclo}(\Andysh(U)),
  U \in \fwp(\Var), T \neq \emptyset \Rightarrow \emptyset \in T \bigr \}
  \enspace ,
\]
where $\partof_{\downclo}(\Andysh(U))$ is the powerset of downward closed
subsets of $\Andysh(U)$ according to $\leq$ and $[T_1]_{U_1} \leq_\an
  [T_2]_{U_2}$ iff $U_1=U_2$ and $T_1 \subseteq T_2$.
For instance, the set $\{xy^\infty z\}$ is not downward closed, while
$\{xyz,xy^\infty z\}$ is downward closed.
There is a Galois
insertion of $\ShLinp$ into $\Linp$ given by the pair of adjoint maps
$\gamma_\an: \ShLinp \fun \Linp$  and  $\alpha_\an: \Linp \fun \ShLinp$:
\begin{align*}
  \gamma_\an([T]_U)=\left[ \alpha_\an^{-1}(T) \right]_U
  \qquad
  \alpha_\an([S]_U)=\left[ \downclo \alpha_\an(S) \right]_U
  \enspace ,
\end{align*}
where $\downclo T = \{o \mid o' \in T, o \leq o'\}$ is the downward closure of
$T$. With an abuse of notation, we also apply $\gamma_\an$ and $\alpha_\an$ to
subsets of $\omega$-sharing groups and 2-sharing groups
respectively, by ignoring the set of variables of interest. For instance,
$\gamma_\an([\emptyset, xyz,xy^\infty z]_{x,y,z})= [ \multil\multir, xyz, xy^{2} z,xy^{3} z,xy^{4} z,xy^{5}
  z,\ldots ]_{x,y,z}$.  Moreover, we write $\downclo [S]_U$ as an alternative form for
$[\downclo S]_U$.

\begin{example}
  \label{ex:7}
  Consider the substitution $[\theta]_U = [\{x/s(y,u,y), z/s(u,u), v/u
    \}]_{w,x,y,z}$ in the Example~\ref{ex:2}. Its abstraction in $\ShLinp$ is
  given by
  \[
    \alpha_\an(\alpha_\lp([\theta]_U)) = [\{xy, x^\infty y, xz, xz^\infty,
    w\}]_U = [ \downclo \{x^\infty y, xz^\infty, w \} ]_U \enspace .
  \]
  Analogously, substitutions $[\theta_1]_{U_1} = [\{ x/r(w_1,w_2,w_2,w_3,w_3),
        y/a, z/r(w_1) \}]_{x,y,z}$ and $[\theta_2]_{U_2} = [\{ x/r(w_4, w_5, w_6, w_8,
        w_8), u/r(w_4,w_7), v/r(w_7,w_8)\}]_{u,v,x}$ from Example~\ref{ex:3} are
  abstracted into $[T_1]_{U_1}=\downclo[x^\infty, xz]_{U_1}$ and $[T_2]_{U_2} =
    \downclo[uv, ux, vx^\infty, x]_{U_2}$  respectively.
\end{example}

\subsection{Matching operator}

By composition, $\alpha_\an \circ \alpha_\lp$ is an abstraction from
$\Isubst_\sim$ to $\ShLinp$. Properties of the Galois connection $\langle
  \alpha_\an, \gamma_\an \rangle$ may be lifted to properties of  $\alpha_\an
  \circ \alpha_\lp$. In our case, we need to define an abstract matching
$\matchan$ over $\ShLinp$ which is optimal w.r.t.~$\match$. However,
optimality of $\match_\an$ w.r.t.~$\match$ is immediately derived by optimality
w.r.t.~$\matchlp$. Since the correspondence between $\Linp$ and $\ShLinp$ is
straightforward, the same happens for $\matchlp$ and $\matchan$.

\begin{definition}[Matching over $\ShLinp$]
  \label{def:matchan}
  Given $[T_1]_{U_1}, [T_2]_{U_2} \in \ShLinp$, we define
  \[
    \matchan([T_1]_{U_1},[T_2]_{U_2})=\bigl[T'_2 \cup \downclo \bigl\{ o \in
    \Andysh(\Var) \mid
    o_{|U_1} \in T_1 \wedge o_{|U_2} \in (T''_2)^* \bigr\}
    \bigr]_{U_1 \cup U_2}
  \]
  where $T'_2=\{ B \in T_2 \mid B_{|U_1} = \emptyset \}$, $T''_2 = T_2
    \setminus T'_2$ and $T^* = \{ \biguplus X \mid X \subseteq T \cup T^2 \}$.
\end{definition}

\begin{example}
  \label{ex:anopt}
  Under the hypothesis of Examples~\ref{ex:4} and~\ref{ex:7}, and according to the definition of $\matchan$, we have that $T'_2 =  \{ uv \}$, $T''_2 = \{ ux, vx, vx^\infty, x \}$ and
  \begin{multline*}
    \matchan([T_1]_{U_1},[T_2]_{U_2})= \downclo [ uv, u^\infty v^\infty
      x^\infty, uxz, u^\infty x^\infty,
      v^\infty x^\infty, vxz, x^\infty, xz
    ]_{u,v,x,y,z}
    \enspace .
  \end{multline*}
  Note that
  \begin{multline*}
    \alpha_\an(\matchlp([S_1]_{U_1},[S_2]_{U_2})) = \downclo [ uv,
      u^\infty x^\infty, uxz, vx^\infty, x^\infty, xz]_{u,v,x,y,z} \\
    \leq \matchan([T_1]_{U_1},[T_2]_{U_2}) \enspace .
  \end{multline*}
  This is consistent with the fact that $\matchan$ is correct w.r.t.~$\matchlp$. The 2-sharing groups $u^\infty v^\infty x^\infty$, $v^\infty x^\infty$ and $vxz$ do not appear in $\alpha_\an(\matchlp([S_1]_{U_1},[S_2]_{U_2}))$ since  $\matchan$ is not complete w.r.t.~$\matchlp$
\end{example}

As anticipated before, we prove optimality of $\matchan$ w.r.t.~$\matchlp$,
which automatically entails optimality w.r.t.~$\match$.

\begin{theorem}[Correctness and optimality of $\matchan$]
  The operator $\matchan$ is correct and optimal w.r.t.~$\matchlp$.
\end{theorem}
\begin{proof}
  We need to prove that, for each $[T_1]_{U_1}, [T_2]_{U_2} \in \ShLinp$,
  \begin{equation}
    \label{eq:optan1}
    \matchan([T_1]_{U_1},[T_2]_{U_2})=\alpha_\an(\matchlp(\gamma_\an([T_1]_{U_1}),
      \gamma_\an([T_2]_{U_2}))) \enspace .
  \end{equation}
  To ease notation, we denote $\gamma_\an([T_1]_{U_1})$ and
  $\gamma_\an([T_2]_{U_2})$ by $[S_1]_{U_1}$ and $[S_2]_{U_2}$ respectively.
  Moreover, we denote with $S'_2$, $S''_2$, $T'_2$ and $T''_2$ the subsets of
  $S_2$ and $T_2$ given accordingly to Definitions~\ref{def:matchlp} and
  \ref{def:matchan}. Since $\supp{B}=\supp{\alpha_\an(B)}$, given $B \in S_2$
  we have that $B \in S_2'$ iff $\alpha_\an(B) \in T_2'$.

  Let $o \in \matchan([T_1]_{U_1},[T_2]_{U_2})$. If $o \in T'_2$, consider
  any $B \in \alpha^{-1}_\an(o) \subseteq S_2$. Then, $B \in S'_2 \subseteq
    \matchlp([S_1]_{U_1},[S_2]_{U_2})$. Therefore, $o = \alpha_\an(B)$ is in the
  right-hand side of \eqref{eq:optan1}. If $o \notin T'_2$ then $o \leq o'$
  where $o'_{|U_1} \in T_1$ and $o'_{|U_2} \in (T''_2)^*$, \ie there is $X
    \subseteq T''_2 \cup (T''_2)^2$ such that $o'_{|U_2}=\biguplus X$. For each $o'' \in
    X \cap T_2$,  let $B_{o''} \in  \alpha^{-1}_\an(o'')$. For each $o'' \in X
    \setminus T_2$, we have $o'' = (o''')^2$ for some $o''' \in T_2$, and let
  $B_{o''} \in \alpha^{-1}_\an(o''')$.  Let $\mathcal X$
  be a multiset containing a single copy of each $B''_o$ for $o \in X \cap T_2$
  and two copies of $B''_o$ for each $o \in X \setminus T_2$, e.g.,
  $\mathcal X = \multil B_{o''} \mid
    o'' \in X \cap T_2 \multir \uplus \biguplus \multil \multil B_{o''}, B_{o''}
    \multir  \mid o'' \in X \setminus T_2 \multir)$. Note that
  $\alpha_\an(\biguplus \mathcal X)=\biguplus \alpha_2(X)=o'_{|U_2}$ by
  \eqref{eq:abstraction1} of Proposition~\ref{prop:abstraction2}. Then, consider the  $\omega$-sharing group $C$ such
  that
  \[
    C(v)=\begin{cases}
      (\biguplus \mathcal X)(v) & \text{for $v \in U_2$,}                               \\
      o'(v)                     & \text{if $v \in U_1 \setminus U_2$ and $o(v)\leq 1$,} \\
      2                         & \text{otherwise}.
    \end{cases}
  \]
  It is clear that $\alpha_\an(C_{|U_1})= o'_{|U_1}$, hence $C_{|U_1} \in S_1$.
  Moreover, $C_{|U_2} = \Multisum X$ with $\mathcal X \in \mwp(S_2'')$.
  Therefore, we have $C  \in \matchlp([T_1]_{U_1}, [T_2]_{U_2})$ and
  $\alpha_\an(C)=o'$ is in the
  right-hand side of \eqref{eq:optan1}. The same holds for $o$ by downward
  closure of $\alpha_\an$.

  Conversely, let $o \in \alpha_\an(\matchlp([S_1]_{U_1}, [S_2]_{U_2}))$.
  As a consequence, there exists $B \in \matchlp([S_1]_{U_1}, [S_2]_{U_2})$
  such that $o \leq o'$ and $o' = \alpha_\an(B)$. It is enough to prove that
  $o' \in \matchan([T_1]_{U_1},[T_2]_{U_2})$. If $B \in S'_2$ then $o' \in
    T'_2$, hence  $o' \in \matchan([T_1]_{U_1},[T_2]_{U_2})$. If $B \notin
    S'_2$, then $B_{|U_1} \in S_1$ and $B_{|U_2} = \biguplus \mathcal X$ where
  $\mathcal X  \in \mwp(S''_2)$. By $B_{|U_1} \in S_1$ we have $o'_{|U_1} =
    \alpha_\an(B_{|U_1}) \in T_1$. For each $C \in \mathcal X$  with $\mathcal
    X(C)=1$, we define $o_{C}=\alpha_\an(C) \in T''_2$, while for each $C$ with
  $\mathcal X(C)>1$, we define $o_{C}=\alpha_\an(C)^2 \in (T''_2)^2$. We have
  that $o''_{|U_2} = \alpha_\an(B)_{|U_2} = \alpha_\an(B_{|U_2})=\alpha_\an(\Multisum \mathcal
    X)=\Multisum X$ where $X = \{ o_C \mid C \in \mathcal X\}$ is an element of
  $\partof(T''_2 \cup (T''_2)^2)$. This means that  $o' \in
    \matchan([T_1]_{U_1},[T_2]_{U_2})$.
\end{proof}

Although $\matchan$ is not complete w.r.t.~either $\matchlp$ or $\match$, we claim it
enjoys a property analogous to the one in Lemma~\ref{lem:matchcompl}.

\subsection{Optimization}
\label{sec:optimization}

In a real implementation, we would like to encode an element of $\ShLinp$ with
the set of its maximal elements. This works well only if we may compute
$\matchan$ starting from its maximal elements, without
implicitly computing the downward closure. We would also like $\matchan$ to
compute as few non-maximal elements as possible. We provide a new algorithm
for $\matchan$ following this approach.

%

\begin{definition}
  \label{def:matchopt}
  Given $T_1, T_2$ sets of $2$-sharing groups and $U_1, U_2 \subseteq \Var$, we
  define
  \[
    \matchaprime(T_1,U_1,T_2,U_2) = T'_2 \cup \bigcup_{o \in T_1}
    \matchaprime(o) \enspace ,
  \]
  where
  \[
    \matchaprime(o) = \left\{ \left(o \wedge \Multisum X\right) \uplus \Multisum (X \cap \overbar
    T)
    \Bigm| X \subseteq T''_2, \left(\Multisum \supp{X}\right)_{|U_1} \leq o_{|U_2} \right\} \enspace,
  \]
  with $T'_2$ and $T''_2$ as in Definition~\ref{def:matchan}, $\overbar T = \{ o' \in T''_2 \mid \forall v \in o' \cap U_1, o(v)=\infty \}$  and
  \[
    o \wedge o'=  \lambda v.\begin{cases}
      o(v)             & \text{if $v \in U_1 \setminus U_2$}, \\
      min(o(v), o'(v)) & \text{if $v \in U_1 \cap U_2$,}      \\
      o'(v)            & \text{otherwise.}
    \end{cases}
  \]
\end{definition}

The operator $\matchaprime$ aims at computing the set of maximal $2$-sharing
groups in $\matchan([\downclo T_1]_{U_1},[\downclo T_2]_{U_2})$. It works by
considering one sharing group $o \in T_1$ at a time, and calling an auxiliary
operator which computes the maximal $2$-sharing groups compatible with $o$. A
$2$-sharing group $o'$ is compatible with $o$ if $\supp{o'}\cap U_1 = \supp{o}
  \cap U_1$.

When computing the auxiliary operator $\matchaprime(o)$ we choose a
subset  $X$ of $T''_2$. Given $o' \in X$, note that $\supp{o'}$ may be viewed
as a $2$-sharing group which is the linearized version of $o'$: it has the
same support as $o'$, but all variables are linear. If $o' \in T''_2$, its
linearization is in $\downclo T_2$. The choice of $X$ is valid if the multiset
sum of all its linearizations is smaller than $o$ for all the variables in
$U_1 \cap U_2$. Once established that $X$ is a valid
choice, we do not take directly  $\Multisum \supp{X}$ as the resulting
$2$-sharing group, but we try to find an $o'' \geq \Multisum X$.

To this purpose, we observe that, given $o' \in X$, if $o(v)=\infty$
for each $v \in \supp{o'} \cap U_1$, then we may take $o'$ twice. We denote with
$\overbar T$  the set of all the sharing groups in  $T''_2$ which such a
property, and we unconditionally add to the result the sharing groups
in $X \cap \overbar{T}$, so that all these sharing groups are taken twice.
Therefore, the biggest element
compatible with $o$ is $(o \wedge \Multisum X) \uplus \Multisum (X \cap
  \overbar T)$.

\begin{example}
  Under the hypothesis of Examples~\ref{ex:4} and~\ref{ex:7}, let us define
  $T_1=\{x^\infty, xz\}$ and $T_2 = \{ uv, ux, vx^\infty, x \}$. We have
  $T'_2  = \{uv\}$ and $T''_2 = \{ux, vx^\infty, x \}$.

  Let us compute $\matchaprime(x^\infty)$. We have $\overbar{T}
    = T''_2$. If we take $X=\{ux, xv^\infty\}$, we have $\left(\Multisum \supp{X}\right)_{|U_1} =
    x^\infty$, hence the choice is valid. The corresponding result is $(x^\infty
    \wedge \Multisum X) \uplus \Multisum X = uv^\infty x^\infty \uplus
    uv^\infty x^\infty = u^\infty v^\infty x^\infty$. Overall, we have
  $\matchaprime(x^\infty)= \{ u^\infty x^\infty, v^\infty x^\infty,
    x^\infty, u^\infty x^\infty v^\infty \}$. Note that some results are
  computed by different choices of $X$. For example, both $\{ux, x\}$ and
  $\{ux\}$ generates $u^\infty x^\infty$.

  Let us compute $\matchaprime(xz)$ and take $X=\{ v x^\infty
    \}$. We have that $\left(\Multisum \supp{X}\right)_{|U_1}  = x = xz_{|U_2}$. In this case,
  $\overbar{T''_2} = \emptyset$. Therefore, the result is $x \wedge
    vx^\infty = vx$. Note that if we take $X= \{xz, x\}$ then $\left(\Multisum
    \supp{X}\right)_{|U_1} = x \uplus x= x^\infty$, and the choice is not valid. This shows
  that, due to the downward closure, choosing in $X$ a sharing group with a non-linear variable is very different from choosing the same variable twice. At the end,
  we have $\matchaprime(xz) = \{ uxz, vxz, xz \}$. Finally,
  \[
    \matchan(T_1,U_1,T_2,U_2) = \{ uv,  u^\infty x^\infty, v^\infty x^\infty,
    x^\infty, u^\infty x^\infty v^\infty , uxz, vxz, xz \} \enspace .
  \]
  This is exactly the set of maximal elements in $\matchan([\downclo
      T_1]_{U_1}, [\downclo T_2]_{U_2})$.
\end{example}

We can show that the correspondence between $\matchan$ and $\matchaprime$ in
the previous example was not by chance, proving the following:
%
%
\begin{theorem}
  Given $T_1, T_2$ sets of $2$-sharing groups such that $\downclo [T_1]_{U_1},
    \downclo [T_2]_{U_2} \in \ShLinp$, we have
  \[
    \matchan(\downclo[T_1]_{U_1},\downclo [T_2]_{U_2}) = \downclo
    [\matchaprime(T_1,U_1,T_2,U_2) ]_{U_1 \cup U_2} \enspace.
  \]
\end{theorem}
\begin{proof}
  We start by proving that if $o$ is an element in $\matchan(\downclo[T_1]_{U_1},\downclo [T_2]_{U_2})$, then $o \in \downclo \matchaprime(T_1,U_1,T_2,U_2)$.  If $o \in \downclo T'_2$ this is trivial since $T'_2 \subseteq \matchaprime(T_1,U_1,T_2,U_2)$. Otherwise, $o_{|U_1} \in \downclo T_1$ and $o_{|U_2} \in (\downclo T''_2)^*$, according to Definition~\ref{def:matchan}. Let $o_1  \in T_1$ such that $o_1 \geq o_{|U_1}$, we want to prove that there exists $\bar o \in \match'_\an(o_1)$ such that $o \leq \bar o$.

  Since $o_{|U_2} \in  (\downclo T''_2)^*$, there are $X_a \subseteq \downclo  T''_2$ and $X_b \subseteq (\downclo T''_2)^2$ such that $o_{|U_2} = \Multisum X_a \multisum \Multisum X_b$. For each $o_a \in X_a$, consider $o'_a \in T''_2$ such that $o'_a \geq o_a \in T''_2$. Let $Y_a$ be the set of all those $o'_a$. For each $o_b \in X_b$, consider an $o'_b \in T''_2$ such that $(o'_b)^2 = o_b$. Let $Y_b$ be the set of all those $o'_b$. By construction $\Multisum Y_a \geq \Multisum X_a \geq \Multisum \supp{Y_a}$ and $\Multisum Y_b \multisum \Multisum Y_b = \Multisum \supp{Y_b} \multisum \Multisum \supp{Y_b} = \Multisum X_b$.

  Let $Y= Y_a \cup Y_b$. Obviously $Y \subseteq T''_2$ and $(\Multisum \supp{Y})_{|U_1} = \Multisum \supp{Y_a} \multisum \Multisum \supp{Y_b} \leq (\Multisum X_a \multisum \Multisum X_b)_{|U_1} = (o_{|U_2})_{|U_1} = (o_{|U_1})_{|U_2} \leq (o_1)_{|U_2}$. Therefore, $\bar o = (o_1 \wedge \Multisum Y) \multisum \Multisum (Y \cap \overbar{T}) \in \matchan'(o_1)$. We now prove that $o \leq \bar o$.

  Given any $o_b \in X_b$, we have that $o_1(v) = o_b(v)=\infty$ for each  $v \in o_b \cap U_1$. This implies that $Y_b \subseteq \overbar{T}$.  Therefore, $\bar o \geq (o_1 \wedge \Multisum Y_a \multisum \Multisum Y_b) \multisum \Multisum Y_b = (o_1 \wedge \Multisum Y_a) \multisum \Multisum Y_b \multisum \Multisum Y_b  \geq (o_1 \wedge \Multisum X_a) \multisum \Multisum X_b = o_1 \wedge (\Multisum X_a \multisum \Multisum X_b) = o_1 \wedge o_{|U_2}$. Now, if $v  \in U_1 \setminus U_2$, we have $\bar o(v)= o_1(v) \geq o(v)$. If $v \in U_2 \setminus U_1$ we have $\bar o(v) \geq o_{|U_2}(v) = o(v)$. Finally, if $v \in U_1 \cap U_2$, then $o_1(v) \geq o_{|U_1}(v) = o_{|U_2}(v)=o(v)$, hence $\bar o(v) \geq o(v)$. This concludes one side of the proof.

  For the other side of the equality, assume $o \in \match'_2(T_1, U_1, T_2, U_2)$ and prove $o \in \matchan(\downclo[T_1]_{U_1},\downclo [T_2]_{U_2})$.  If $o \in T'_2$ this is trivial, since $\matchan(\downclo[T_1]_{U_1},\downclo [T_2]_{U_2}) \supseteq \downclo T'_2$. Otherwise, $o = (o_1 \wedge \Multisum X) \multisum \Multisum (X \cap \overbar{T})$ for some $o_1 \in T_1$ and $X \subseteq T''_2$ satisfying the properties of Definition~\ref{def:matchopt}. Consider $Y=(X \setminus \overbar T) \cup (X \cap \overbar T)^2$ which is an element of $T''_2 \cup (T''_2)^2$ such that $\Multisum Y = \Multisum X \multisum \Multisum (X \cap \overbar T) \in (T''_2)^*$. It is enough to prove that $o_{|U_1}=o_1$ and $o_{|U_2}= \Multisum Y$. If $v \in U_1 \setminus U_2$ we have $o(v)=o_1(v)$. Note that if $\bar o \in X \cap \overbar T$ and $v \in \bar o \cap U_1$, then $o_1(v)=\infty$. Therefore, for each $v \in U_1 \cap U_2$, we have $o(v)=o_1(v)$. Finally, if $v \in U_2$, we have $o(v)=\Multisum X \multisum \Multisum (X \cap \overbar T)= \Multisum Y$.
\end{proof}

\section[Abstract matching over Sharing X Lin]{Abstract matching over $\Sharing \times \Lin$}
\label{sec:shlin}

The reduced product $\ShLin=\Sharing \times \Lin$  has been used
for a long time in the analysis of aliasing properties, since it was
recognized quite early that the precision of these analyses could be greatly
improved by keeping track of the linear variables. In the following, we briefly recall the definition of the abstract domain following the presentation in \cite{AmatoS10-tplp}.

\subsection[The domain ShLin]{The domain $\ShLin$}

The domain $\ShLin$ couples an object of $\Sharing$ with the set of variables
known to be linear. Each element of $\ShLin$ is
therefore a triple: the  first component is an object of $\Sharing$, the
second component is an object  of $\Lin$, that is, the set of variables which
are linear in all the sharing
groups of the first component, and the third component is the set of variables
of interest. It is immediate that $\ShLin$ is an abstraction of $\ShLinp$ (and
thus of  $\Linp$).
\[
  \ShLin = \{ [S,L,U]\mid S \subseteq\partof(U), (S \neq\emptyset \Ra
  \emptyset \in S), L \supseteq U \setminus \vars(S), U \in \fwp(\Var)
  \} \enspace ,
\]
with the approximation relation $\leq_\shl$ defined as $[S,L,U]
  \leq_\shl [S',L',U']$ iff $U=U'$, $S \subseteq S'$, $L \supseteq L'$.
There is a Galois insertion of $\ShLin$ into $\ShLinp$ given by the
pair of maps:
\begin{align*}
  \alpha_\shl([T]_U)   & =[\{ \supp{o} \mid o \in T \},\{ x \in U
  \mid \forall o \in T.\ o(x) \leq 1\},U] \enspace ,                     \\
  \gamma_\shl([S,L,U]) & =[\downclo \{ B_L \mid B \in S \}]_U \enspace , \\
  \intertext{where $B_L$ is the 2-sharing group which has the same support of
    $B$, with linear variables dictated by the set $L$:}
  B_L                  & =\lambda v \in \Var. \begin{cases}
    \infty & \text{if $v \in B \setminus L$,} \\
    1      & \text{if $v \in B \cap L$,}      \\
    0      & \text{otherwise.}
  \end{cases}
\end{align*}

The functional composition of $\alpha_\lp$, $\alpha_\an$ and $\alpha_\shl$
gives
the standard abstraction map from substitutions to $\ShLin$.
We still use the
polynomial notation to represent sharing
groups, but now all the exponents are fixed to one.  Note that the
last component $U$ in $[S,L,U]$ is redundant since it can be retrieved
as $L \cup \vars(S)$.  This is because the set $L$ contains all the
ground variables.

\begin{example}
  \label{ex:9}
  Consider the substitution $[\theta]_U = [\{x/s(y,u,y), z/s(u,u), v/u
    \}]_{w,x,y,z}$ in the Example~\ref{ex:2}. Its abstraction in $\ShLin$ is
  given by
  \[
    \alpha_\shl(\alpha_\an(\alpha_\lp([\theta]_U))) = [ \{ xy, xz, w\}, \{ y, w \},
    U ] \enspace .
  \]
  Analogously, substitutions $[\theta_1]_{U_1} = [\{ x/r(w_1,w_2,w_2,w_3,w_3),
        y/a, z/r(w_1) \}]_{x,y,z}$ and $[\theta_2]_{U_2} = [\{ x/r(w_4, w_5, w_6, w_8,
        w_8), u/r(w_4,w_7), v/r(w_7,w_8)\}]_{u,v,x}$ from Example~\ref{ex:3} are
  abstracted into $[S_1,L_1,U_1]=[\{x, xz\},\{y, z\}, U_1]$ and $[S_2, L_2, U_2] =
    [\{ uv, ux, vx, x\}, \{u,v\}, U_2]$ respectively.
\end{example}

\subsection{Matching operator}

We want to provide an optimal abstract matching operator for $\ShLin$. We may
effectively compute $\matchshl$ by composing $\gamma_\shl$, $\matchan$ and
$\alpha_\shl$. However, we provide a more direct characterization of
$\matchshl$, which may potentially improve performance.

First, we define the auxiliary function $nl: \partof(\partof(\Var)) \ra
  \partof(\Var)$
which takes a set $X$ of sharing groups and returns the set of variables which
appear in $X$ more than once. In formulas:
\begin{equation}
  nl(X)= \{ v \in \Var \mid \exists B_1, B_2 \in X, B_1 \neq B_2, v \in B_1 \cap
  B_2 \}
\end{equation}
The name $nl$ stands for non-linear, since it is
used to recover those variables which, after joining sharing groups in $X$,
are definitively not linear.

\begin{definition}[Matching over $\ShLin$ ]
  \label{def:matchshlinopt}
  Given $[S_1,L_1,U_1]$ and $[S_2,L_2,U_2] \in \ShLin$, we define
  \[
    \matchshl([S_1,L_1, U_1],[S_2,L_2,U_2])= \left[s(S'_0 \cup S''_0), l(S'_0 \cup S''_0), U \right],
  \]
  where $U=U_1 \cup U_2$, $S'_2=\{ B \in S_2 \mid B \cap U_1 = \emptyset \}$, $S''_2=S_2 \setminus S'_2$, $\overbar S = \{B \in S''_2 \mid B \cap L_1 = \emptyset \}$ and
  \begin{align*}
    S'_0  & = \bigl\{ \langle B, L_2 \rangle \mid B \in S'_2 \bigr\},                                                                       \\
    S''_0 & = \Big\{ \left\langle B \cup \bigcup X, L_2  \setminus nl(X) \setminus \bigcup (X \cap \overbar S) \right\rangle \Bigm| B \in S_1, \\
          & \qquad X \subseteq S''_2,  B \cap U_2 = \left(\bigcup X\right) \cap U_1, L_1 \cap nl(X)=\emptyset \Big\},                      \\
    s(H)  & = \left\{ B \mid \langle B, L \rangle \in H \right\},                                                                          \\
    l(H)  & = \bigcap \left\{ L_1 \cup L \cup (U \setminus B) \mid \langle
    B,L \rangle \in H \right\}.
  \end{align*}
\end{definition}


This operator is more complex than previous ones, because linearity
information is not connected to sharing groups and needs to be handled
separately. In view of this and a similar situation which happens for
unification \citep{AmatoS10-tplp}, it seems that the idea of embedding
linearity within sharing groups from \cite{King94-esop} was particularly insightful.

We give an intuitive explanation of $\matchshl$. It essentially works by
simulating the optimized version of $\matchan$ starting from $[T_1]_{U_1}=\gamma_\an([S_1,L_1,U_1])$
and  $[T_2]_{U_2}=\gamma_\an([S_2, L_2, U_2])$. The sets $S'_2$ and $S''_2$
are defined as for the other matching operators. Each element
of $H$, $S'_0$ or $S''_0$ is a pair $\langle B, L \rangle$ which corresponds to
the $2$-sharing group $B_{L \cup L_1}$ in
$\matchan([T_1]_{U_1}, [T_2]_{U_2})$. The component $B$ is the support of
the $2$-sharing group, while $L$ is the set of linear variables.  We only
record a subset of the linear variables, since those in $L_1$ are always linear.

The set $S'_0$ encodes all the maximal $2$-sharing groups derived from
$S'_2$. The  set $S''_0$ encodes $2$-sharing groups which may be generated by
gluing the sharing groups in $S''_2$ in a way which is compatible with $S_1$.
A given choice of $X \subseteq T''_2$ is compatible with a sharing group $B \in S_1$
only if $\bigcup X$ and $B$ have the same support on the common variables $U_1 \cap U_2$ and if $X$
does not conflict with the linearity of variables given by $L_1$. This means that we cannot use a
variable in $L_1$ more than once, therefore $L_1 \cap nl(X)=\emptyset$. Note that we may
use a variable $v \in L_1 \setminus L_2$, since $v \notin L_2$ means that $v$ is
possibly, but not definitively, non-linear. On the contrary, if $v \in nl(X)$,
then $v$ is definitively non-linear  in the resultant sharing group. Once established
that $X$ is compatible with $B$, the set $\overbar S$ plays the same role of $\overbar T$ in
Definition~\ref{def:matchopt}: we may join another copy of the sharing groups
in $X \cap \overbar S$, making all their variables non-linear.

Finally, the maps $s$ and $l$ extract from $S'_0$ and $S''_0$ the set of all the sharing
groups, and the set of variables which are linear in all the sharing groups.

\begin{example}
  Consider the substitutions in Example~\ref{ex:9}. According to the definition
  of $\matchshl$, we have: $S'_2 = \{uv\}$, $S''_2 = \overbar S = \{ux, vx, x \}$,
  $S'_0 = \{ \langle uv, \{ u,v \} \rangle \}$ and $S''_0 = \{ \langle ux,
    \emptyset \rangle, \langle vx, \emptyset \rangle, \langle x, \emptyset
    \rangle, \langle uvx, \emptyset \rangle, \langle uxz,
    \emptyset \rangle, \langle vxz, \emptyset \rangle, \langle xz, \emptyset
    \rangle, \langle uvxz, \emptyset \rangle \}$. The final result is
  \begin{multline}
    \matchshl([S_1,L_1,U_1],[S_2,L_2,U_2])=[\{uv, uvx, ux, vx, x, \\
    uvxz, uxz, xz,  vxz \}, \{ y,z\}, U_1  \cup U_2] \enspace .
  \end{multline}
  Considering the results for $\matchan$ in the Example~\ref{ex:anopt}, we have:
  \begin{multline*}
    \alpha_\shl(\matchan([T_1]_{U_1},[T_2]_{U_2}))\\
    = \alpha_\shl(
    \downclo [ uv, u^\infty v^\infty x^\infty, uxz, u^\infty x^\infty,
      v^\infty x^\infty, vxz, x^\infty, xz
    ]_{u,v,x,y,z} )\\
    = [\{ uv, uvx, uxz, ux, vx, vxz, x, xz\}, \{ y,z\},\{u,v,x,y,z\}]
  \end{multline*}
  and $\matchshl([S_1,L_1,U_1],[S_2,L_2,U_2]) >
    \alpha_\shl(\matchan([T_1]_{U_1},[T_2]_{U_2}))$. In particular, the sharing
  group $uvxz$ does not appear in
  $\alpha_\shl(\matchan([T_1]_{U_1},[T_2]_{U_2}))$ which proves that
  $\matchshl$ is not a complete abstraction of $\matchan$. However, the next theorem shows that $\matchshl$ is
  optimal w.r.t.~$\matchan$, and by composition also w.r.t.~$\matchlp$ and $\match$.
\end{example}

\begin{theorem}[Optimality of $\matchshl$]
  The operator $\matchshl$ is correct and optimal w.r.t.~$\matchan$.
\end{theorem}
\begin{proof}
  We prove that, given $[S_1, L_1, U_1]$ and $[S_2, L_2, U_2] \in \ShLin$, we have
  \[
    \matchshl([S_1, L_1, U_1],[S_2, L_2, U_2])= \alpha_\shl(\downclo [\match'_\an(T_1, U_1, T_2, U_2)]_{U})
  \]
  where $T_i = \{B_{L_i} \mid B \in S_i\}$ is the set of maximal elements of $\gamma_\shl([S_i, L_i, U_i])$ and $U=U_1 \cup U_2$.

  Let us define the function $\gamma_0$ which maps a pair $\langle B, L \rangle$ with $B, L \in \partof(\Var)$ to the $2$-sharing group $B_{L \cup L_1}$. We show that, given $H \subseteq \partof(\Var) \times \partof(\Var)$, we have $\alpha_\shl([\downclo \gamma_0(H)]_U) =  [s(H), l(H), U]$. Assume $\alpha_\shl([\downclo \gamma_0(S)]_U])=[R,V,U]$. We have $R=\{\supp{o} \mid o \in \downclo \gamma_0(S)\}= \{\supp{o} \mid o \in \gamma_0(S)\}=\{ \supp{\gamma_o(\langle B,L \rangle)} \mid  \langle B,L \rangle \in S \}=\{ B \mid \langle B,L \rangle \in S \}= s(S)$ and $V=\{ x\in U \mid \forall o \in \downclo \gamma_0(S), o(x) \leq 1\} = \{ x \in U \mid  \forall o \in\gamma_0(S), o(x) \leq 1\} = \{ x \in U \mid  \forall \langle B, L \rangle \in S, B_{L \cup L_1}(x) \leq 1\} \} = \{ x \in U \mid  \forall \langle B, L \rangle \in S, x \in L_1 \cup L \cup (U \setminus B) \} = \bigcap \{ L_1 \cup L \cup (U \setminus B) \mid \langle B, L \rangle \in S \} = l(S)$.

  Therefore, if we prove that $\gamma_0(S'_0 \cup S''_0) = \match'_\an(T_1, U_1, T_2, U_2)$, then we have
  $\match_\shl([S_1,T_1,U_1],[S_2,T_2,U_2])= [s(S'_0 \cup S''_0), l(S'_0 \cup S''_0), U] = \alpha_\shl([\downclo \gamma_0(S'_0 \cup S''_0)]_U]) = \alpha_\shl(\downclo  [\match'_\an(T_1, U_1, T_2, U_2)]_{U}))$.

  Let us take $o \in \gamma_0(S'_0 \cup S''_0)$ and prove $o \in \match'_\an(T_1, U_1, T_2, U_2)$.
  If $o=B_{L \cup L_1}$ for some $\langle B,L \rangle \in S'_0$, then $B \in S'_2$ and $L=L_2$.
  Therefore, $B_{L \cup L_1} = B_{L_2 \cup L_1} = B_{L_2} \in T_2$ since $B \cap U_1 = \emptyset$.
  In particular, $B_{L_2} \in T'_2$, hence $B_{L_2} \in \match'_\an(T_1, U_1, T_2, U_2)$.
  Otherwise, $o= \gamma_0(\langle C, L \rangle)$ for $\langle C, L \rangle \in S''_0$ where $C = B \cup \bigcup X$ and
  $L=L_2 \setminus nl(X) \setminus \bigcup (X \cap \overbar S)$ according to Definition~\ref{def:matchshlinopt}.
  For each sharing group $B' \in X$, consider the $2$-sharing group $B'_{L_2} \in T''_2$ and let $Y$ be the
  set of all those $2$-sharing groups. We want to prove that $o$ is generated by $B_{L_1} \in T_1$ and
  $Y \subseteq T''_2$, according to Definition~\ref{def:matchopt}. First, note that $\supp{(\Multisum
    \supp{Y})_{|U_1}}=\supp{(\Multisum X)_{|U_1}} = (\bigcup X) \cap U_1 = B \cap U_2 = \supp{(B_{L_1})_{|U_2}}$.
  Now, assume $v \in \supp{(B_{L_1})_{|U_2}}$. If $v$ is linear in $(B_{L_1})_{|U_2}$, then $v \in L_1$, hence
  $v \notin nl(X)$ and $v$ is linear in $(\Multisum X)_{|U_1}=(\Multisum	\supp{Y})_{|U_1}$.
  Therefore, $Y$ is a valid choice for $\matchan'$ and $\match'_\an(B)$ generates $o'=(B_{L_1} \wedge \Multisum Y) \multisum \Multisum (Y \cap \overbar T)$. We prove $o'=o$.

  First, we prove that $B' \in \overbar S$ iff $B'_{L_2} \in \overbar{T}$. We have $B'_{L_2} \in \overbar T$ iff $B'_{L_2} \in T''_2$ and $\forall v \in \supp{B'_{L_2}} \cap U_1, B'_{L_1}(v)=\infty$, iff $B' \in S''_2$ and $\forall v \in B' \cap U_1$, $v \notin L_1$ iff $B' \cap L_1 = \emptyset$ iff $B' \in S''_2$ and $B' \cap L_1=\emptyset$ iff $B \in S$.

  Now, it is immediate to check that, $\supp{o}= B \cup \bigcup X =\supp{o'}$. Then, consider $v \in \supp{o'}$. We have that:
  \begin{align*}
     & o'(v)=1 \Leftrightarrow                                                                                           \\
     & (B_{L_1}(v)=1 \vee \left(\Multisum Y\right)(v)=1) \wedge \left(\Multisum (Y \cap \overbar T)\right)=0 \Leftrightarrow \\
     & (v \in L_1 \vee (v \in L_2 \setminus nl(X))) \wedge v \notin \bigcup (X \cap \overbar S) \Leftrightarrow              \\
     & v \in L_1 \vee (v \in L_2 \setminus nl(X) \setminus \bigcup (X \cap \overbar S) \Leftrightarrow                       \\
     & v \in L_1 \vee v \in L \Leftrightarrow                                                                            \\
     & B_{L \cup L_1}(v) = o(v)=1
  \end{align*}
  %
  %

  It remains to prove that given  $o \in \match'_\an(T_1, U_1, T_2, U_2)$, we have $o \in \gamma_0(S'_0 \cup S''_0)$. If $o \in T'_2$, then $o=B_{L_2}$ with $B \in S'_2$. Therefore, $\langle B, L_2 \rangle \in S'_0$ and $o= \gamma_0(\langle B,L_2 \rangle)$. Otherwise, $o=(o' \wedge \Multisum Y) \multisum \Multisum (Y \cap \overbar T)$.  We know $o'=B'_{L_1}$ with $B' \in S_1$ and let $X=\{ \supp{o''} \mid o'' \in Y \}$. By Definition \ref{def:matchopt}, $\bigcup X \cap U_1 = \supp{\multisum Y} \cap U_1 = \supp{o'} \cap U_2 = B' \cap U_2$. Moreover, if $v \in L_1$ then $o'(v) \leq 1$ and therefore $v$ cannot appear twice in $Y$, which means $v \notin nl(X)$, hence $L_1 \cap nl(X) = \emptyset$. Therefore, $B'$ and $X$  makes a valid choice for  $\match_\shl$ and generate the pair
  \[
    \langle B, L \rangle = \left\langle B' \cup \bigcup X, L_2 \setminus nl(X) \setminus \bigcup (X \cap \overbar S) \right\rangle .
  \]
  Using the first half of the proof, it is easy to check that $B_{L \cup L'} = o$, which terminates the proof.
  %
  %
  %
\end{proof}

\section{Evaluation of matching in goal-dependent analysis}
\label{sec:examples}

We now show two examples, in the context of goal-dependent analysis, where the newly introduced matching operators improve the precision \wrt what is attainable using only the mgu operators in \cite{AmatoS10-tplp}.

\subsection{Matching and backward unification}

First, we show with a simple example how the matching operator strictly improves the result of a standard goal-dependent analysis using forward and backward unification. Consider the goal $p(x,f(x,z),z)$ with
the (abstract) call substitution $[x,z]_{xz}$ and the trivial program with just one clause:

\medskip\noindent
\hspace{1cm} $p(u,v,w).$
\medskip

In order to analyze the goal, we first need to perform the forward unification between the call substitution  $[x,z]_{xz}$, the goal $p(x,f(x,z),z)$ and the head of the clause $p(u,v,w)$, and then project the result on the variables of the clause.
In order to keep the notation simple, we do not perform renaming, unless necessary. This amounts to computing:
\[
	\mgu_\lp([x,z]_{xz}, \{u/x, v/f(x,z), w/z\}) = [uvx, vwz]_{uvwxz} \enspace .
\]
By projecting the result on the variables of the clause we obtain the  entry substitution $[uv, vw]_{uvw}$. Since the clause has no body, it is immediate to see that the  exit substitution coincides with the entry substitution.

We now need to compute the backward unification of the exit substitution $[uv, vw]_{uvw}$, the call substitution $[x, z]_{xz}$ and $\theta=\{u/x, v/f(x,z), w/z\}$. If we implement this operator with the aid of matching, we may first unify  $[x, z]_{xz}$  with $\theta$, obtaining $[uvx, vwz]_{uvwxz}$ as above, and then apply the matching with the exit substitution $[uv, vw]_{uvw}$ obtaining:
\[
	\match_\lp([uv, vw]_{uvw}, [uvx, vwz]_{uvwxz} ) = [uvx, vwz]_{uvwxz} \enspace .
\]
By projecting the result on the variables of the goal we obtain the result $[x, z]_{xz}$ which proves that $x$ and $z$ do not share.

On the contrary, if we avoid matching, the backward unification operator must unify  $[x, z]_{xz}$ with $[uv, vw]_{uvw}$ and $\theta$ in any order it deems fit.
 However, this means that the operator must correctly approximate the mgu of any substitution $\theta_1$ in the concretization of the call substitution  $[x,z]_{xz}$, with any substitution $\theta_2$ in the concretization of the exit substitution $[uv, vw]_{uvw}$ and $\theta=\{u/x, v/f(x,z), w/z\}$.
%
%
If we choose $\theta_1 = \epsilon$ and $\theta_2 = \{v/f(w,u)\}$, we obtain that the unification of $\theta_1$, $\theta_2$ and $\theta$ is
\[
	\{u/x, v/f(x,x), w/x, z/x \}
\]
which is not in the concretization of $[uvx, vwz]_{uvwxz}$, since $u$, $v$ and $w$ share a common variable. Moreover, in this substitution $x$ and $z$ share. This means that, after the projection on the variables of the goals $x$ and $z$, the result will always include the $\omega$-sharing group $xz$.
Note that this consideration holds for any correct abstract unification operator which avoids matching.


It is worth noting that the above example, with minimal changes, also works for $\ShLinp$ and $\ShLin$: also in these cases, matching improves the precision of the analysis.

\subsection[Example on a non-trivial program]{Example on a non-trivial program}

Consider a program implementing the $\member$ predicate. Using the Prolog notation for lists, we have:

\medskip\noindent
\hspace*{1cm}$\member(u, [u|v]).$\\
\hspace*{1cm}$\member(u, [v|w])  \leftarrow \member(u, w).$
\medskip

We want to analyze the goal $\member(x,[y])$ in the domain $\ShLinp$ using the \emph{call substitution} $[xy, xz]_{xyz}$.

We start by considering the first clause of $\member$. The concrete unification of the goal $\member(x,[y])$ and the head of the clause $\member(u, [u|v])$ yields the most general unifier $\theta = \{ x/u, y/u, v/[] \}$. Forward unification computes the entry substitution as the abstract mgu between the call substitution $[xy, xz]_{xyz}$ and $\theta$. Proceeding one binding at a time, we have
\begin{itemize}
  \item $\mgu_\an([xy, xz]_{xyz}, \{x/u\}) = [uxy, uxz]_{uxyz}$;
  \item $\mgu_\an([uxy, uxz]_{uxyz}, \{y/u\}) = \downclo [u^\infty x^\infty y^\infty]_{uxyz}$;
  \item $\mgu_\an(\downclo [u^\infty x^\infty y^\infty]_{uxyz} , \{v/[]\} ) = \downclo [u^\infty x^\infty y^\infty]_{uvxyz}$.
\end{itemize}
Projecting over the variables of the clause, we get the entry substitution $\downclo [u^\infty]_{uv}$. Since this clause has no body, the entry substitution is equal to the exit substitution, and we may proceed to compute the answer substitution trough backward unification.

First, we consider the case when the backward unification is performed using the standard $\mgu_\an$ operator. We need to unify the call substitution $[xy, xz]_{xyz}$, the exit substitution $\downclo [u^\infty]_{uv}$ and the concrete substitution $\theta$ (the same as before). Unifying call and exit substitution is immediate since they are relative to disjoint variables of interest: the result is  $\downclo [u^\infty, xy, xz]_{uvxyz}$, obtained by collecting all sharing groups together. This should be unified with $\theta$. Proceeding one binding at a time, and omitting the set of variables of interest since it does not change, we have:
\begin{itemize}
  \item $\mgu_\an(\downclo [u^\infty, xy, xz], \{x/u\}) = \downclo [u^\infty x^\infty y^\infty, u^\infty x^\infty y^\infty z^\infty, u^\infty x^\infty z^\infty]$;
   \item $\mgu_\an(\downclo [u^\infty x^\infty y^\infty, u^\infty x^\infty y^\infty z^\infty, u^\infty x^\infty z^\infty], \{y/u\}) =  \downclo [u^\infty x^\infty y^\infty, u^\infty x^\infty y^\infty z^\infty]$;
  \item $\mgu_\an( \downclo [u^\infty x^\infty y^\infty, u^\infty x^\infty y^\infty z^\infty],\{v/[]\})=   \downclo [u^\infty x^\infty y^\infty, u^\infty x^\infty y^\infty z^\infty]$.
\end{itemize}
Projecting over the set of variables in the goal we get $\downclo [x^\infty y^\infty, x^\infty y^\infty z^\infty]_{xyz}$.

On the contrary, if we perform backward unification using the matching operation, we need to compute the matching of the exit substitution $[u^\infty]_{uv}$ with the entry substitution before variable projection $[u^\infty x^\infty y^\infty]_{uvxyz}$, namely:
\[
  \match_\an(\downclo [u^\infty]_{uv}, \downclo [u^\infty x^\infty y^\infty]_{uvxyz}) =   \downclo [u^\infty x^\infty y^\infty]_{uvxyz} \enspace .
\]
Projecting over the variables of the goal, we get $\downclo [x^\infty y^\infty]_{xyz}$: using matching we can prove that $z$ is ground in the answer substitution.

However, we still need to check what happens when we analyze the second clause of the $\member$ predicate. In this case, the concrete unification between $\member(x, [y])$ and $\member(u, [v|w])$ gives the substitution $\theta = \{x/u, y/v, w/[]\}$. Then, forward unification between the call substitution and $\theta$ gives
\begin{itemize}
  \item $\mgu_\an([xy, xz]_{xyz}, \{x/u\}) = [uxy, uxz]_{uxyz}$;
  \item $\mgu_\an([uxy, uxz]_{uxyz}, \{y/v\}) = [uvxy, uxz]_{uvxyz}$;
  \item $\mgu_\an([uvxy, uxz]_{uvxyz}, \{w/[]\}) = [uvxy, uxz]_{uvwxyz}$.
\end{itemize}
Projecting over the variables of the clause we get the entry substitution $[uv, v]_{uvw}$.

Now we should compute the answer substitution of the body $\member(u, w)$ under the call substitution $[uv, u]_{uvw}$. We could proceed showing all the details, but we try to be more concise. In the abstract substitution $[uv, w]_{uvw}$ the variable $w$ is known to be ground. When $\member$ is called with its second argument ground, the first argument becomes ground too. This property is easily captured by $\Sharing$ and more precise domains, independently of the fact that we use matching or not for the backward unification.
Therefore, we can conclude that the answer substitution for the goal $\member(u, w)$ under the entry substitution $[uv, u]_{uvw}$ is $[\emptyset]_{uvw}$. Although we do not generally write the empty $\an$-sharing group $\emptyset$ in an element of $\ShLinp$, in this case it is important to write it in order to distinguish $[\emptyset]_{uvw}$, denoting those substitutions in which $u, v, w$ are ground, from $[]_{uvw}$, denoting a non-succeeding derivation.

Performing the backward unification of $[\emptyset]_{uvw}$, we get the answer substitution $[\emptyset]_{xyz}$, independently of the use of matching.

Putting together the results we got for analyzing the goal $\member(x, [y])$ according to the two clauses of the program, we have shown that using matching we get $[x^\infty y^\infty]_{xyz}$ while using standard unification we get $[x^\infty y^\infty, x^\infty y^\infty z^\infty]_{xyz}$, and we are not able to prove that $z$ is ground after the goal returns.

\section{Conclusion}
\label{sec:conclusion}

In this paper we have extended the domain $\Linp$ \citep{AmatoS10-tplp} to goal-dependent analysis, by introducing a matching operator, and proved its optimality. From this operator we have derived the optimal matching operators for the well known $\ShLinp$ \citep{King94-esop} and $\Sharing \times \Lin$ \citep{MuthukumarH92-jlp} abstract domains.

As far as we know, this is the first paper which shows matching optimality results for domains combining sharing and linearity information. In particular, the matching operators presented in \citep{HansW92-tr, King00-jlp} for the domain \textit{SFL}, which combines set-sharing, linearity and freeness information, are not optimal, as shown by \cite{AmatoS09-tplp}.

Recently logic programming has been used as an intermediate representation for
the analysis of imperative or object-oriented programs and services (see, e.g., \cite{PeraltaGS98-sas,HenriksenG09-scam,BentonF07-ppdp,MendezLojoNH07-lopstr,SpotoMP10-toplas,AlbertAGPZ12-tcs,IvanovicCH13-computing,GangeNSSS15-tplp,DeAngelisFGH+21-tplp}).
Since many of these approaches
use existing logic program analysis on the transformed program, we believe that
they can benefit from more
precise logic program analysis.

%
%


\medskip

\noindent Competing interests: The authors declare none.

\bibliographystyle{tlplike}
\bibliography{matching}

\end{document}